\documentclass[paper=letter,DIV=10,11pt]{scrartcl} %

\usepackage[utf8]{inputenc} 	%
\usepackage[T1]{fontenc} 		%
\usepackage{csquotes} 			%
\usepackage{textcomp}			%

\usepackage[tt=false,type1=true]{libertine}	%
\usepackage[libertine]{newtxmath}				%
\usepackage{microtype}

\usepackage{amsthm}				%

\usepackage{mathtools}			%
\usepackage{bm} 					%
\usepackage{bbm}
\usepackage{mathdots}
\usepackage{enumitem} 			%

\usepackage[most]{tcolorbox}
\usepackage{tabularx}
\usepackage{array}
\usepackage{colortbl}
\usepackage{multirow}
\tcbuselibrary{skins}
\tcbuselibrary{breakable}

\usepackage{varwidth}
\usepackage{subcaption}

\makeatletter
\tcbset{
  tabular/.style={
    boxsep=\z@,top=\z@,bottom=\z@,leftupper=\z@,rightupper=\z@,
    toptitle=1mm,bottomtitle=1mm,boxrule=0.25mm,
    before upper={\arrayrulecolor{black}%
      \tcb@hack@currenvir\tabular{#1}},
    after upper=\endtabular\arrayrulecolor{black}},
}
\makeatother
\newcolumntype{Y}{>{\raggedleft\arraybackslash}X}
\tcbset{fancy table/.style={%
		enhanced,fonttitle=\bfseries\small,fontupper=\small,
		colback=white,colframe=black,
		colbacktitle=gray,coltitle=white,center title,
		arc=0mm,sharpish corners
}}

\usepackage{xcolor} 				%
\colorlet{verylightgray}{lightgray!50!white}

\usepackage{pgffor} 				%

\usepackage{framed} 				%

\usepackage{float}				%
\usepackage{booktabs} 			%

\usepackage{hyperref}			%
\usepackage[noabbrev,capitalise]{cleveref} %

\allowdisplaybreaks

\numberwithin{equation}{section}

\def\scrcorr#1#2{%
	\def\subcorrlength{#1}%
	\def\supcorrlength{#2}%
	\futurelet\next\subcorrcheck%
}
\def\subcorrcheck{\ifx\next_\expandafter\subcorr\else\expandafter\supcorrcheck\fi}
\def\subcorr_#1{%
	_{\mkern\subcorrlength#1}%
	\futurelet\next\supcorrcheck%
}
\def\supcorrcheck{\ifx\next^\expandafter\supcorr\fi}
\def\supcorr^#1{%
	^{\mkern\supcorrlength#1}%
}
\newcommand*{\bmskew}{\skew4}

\newcommand*	{\?}			{\ensuremath{\:\cdot\:}} 	%
\let\oldtop\top				%
\renewcommand*	{\top}		{\oldtop\scrcorr{-4mu}{0mu}}
\let\from\colon				%
\let\with\colon				%
\let\epsilon\varepsilon		%
\let\phi\varphi				%
\renewcommand*	{\restriction}	{\mathord{\upharpoonright}}	
\let\oldexists\exists
\let\oldforall\forall
\renewcommand*{\exists}{\mathord{\oldexists\mkern2mu}}
\renewcommand*{\forall}{\mathord{\oldforall\mkern2mu}}

\let\mathbb\varmathbb		%
\let\bb\mathbb					%
\let\fr\mathfrak				%
\let\cal\mathcal				%
\let\sf\mathsf					%

\let\oldH\H
\makeatletter
\foreach \x in {A,...,Z}{%
	\expandafter\xdef \csname \x \endcsname			{ \noexpand\cal{\x} }%
	\expandafter\xdef \csname \x\x \endcsname			{ \noexpand\bb{\x} }%
	\expandafter\xdef \csname \x\x\x \endcsname		{ \noexpand\fr{\x} }%
	\expandafter\xdef \csname \x\x\x\x \endcsname	{ \noexpand\bm{\x} }%
}%
\makeatother

\DeclarePairedDelimiter	{\set}	{\lbrace}	{\rbrace}	%
\DeclarePairedDelimiter	{\event}	{\lbrace}	{\rbrace}	%
\DeclarePairedDelimiter	{\spn}	{\langle}	{\rangle} 	%
\DeclarePairedDelimiter	{\size}	{\lvert}		{\rvert}		%
\DeclarePairedDelimiter	{\abs}	{\lvert}		{\rvert}		%
\DeclarePairedDelimiter	{\pair}	{{\mathord{<}}}	
															{{\mathord{>}}}	%

\makeatletter
\newcommand*	{\dummydelim}	[2]	{$\left#1\vphantom{#2}\right.$}
\newcommand		{\nbrace}		[2]	{\sbox0{\dummydelim{#1}{#2}}\hspace{\the\dimexpr -0.85\wd0 + 2pt\relax}}
\DeclarePairedDelimiterX	{\bag}[1]	{\lbrace}	{\rbrace}
{%
	\nbrace{\lbrace}{#1}\delimsize\lbrace\mathopen{}%
	#1%
	\mathclose{}\delimsize\rbrace\nbrace{\rbrace}{#1}%
}
\makeatother

\newcommand*{\powerset}{\mathcal{P}}	%
\newcommand*{\powerbag}{\mathcal{B}}	%

\newcommand*{\fin}{\ensuremath{\mathsf{fin}}}			%
\newcommand*	{\under}	{\mathbin{\vert}}		%
\DeclareMathOperator	{\Expectation}	{E}		%

\theoremstyle{plain}		%

\newtheorem		{theorem}						{Theorem}		[section]

\newtheorem		{proposition}	[theorem]	{Proposition}

\newtheorem		{lemma}			[theorem]	{Lemma}
\newtheorem		{corollary}		[theorem]	{Corollary}
\newtheorem		{fact}			[theorem]	{Fact}

\theoremstyle{definition} %
\newtheorem		{remark}			[theorem]	{Remark}
\newtheorem		{example}		[theorem]	{Example}
\newtheorem		{definition}	[theorem]	{Definition}

\crefname	{observation}	{Observation}	{Observations}
\crefname	{fact}			{Fact}			{Facts}
\crefname	{notation}		{Notation}		{Notations}
\crefname	{assumption}	{Assumption}	{Assumptions}
\crefname	{equivalence}	{Equivalence}	{Equivalences}
\crefname	{function}		{Function}		{Functions}
\crefname	{condition}		{Condition}		{Conditions}

\crefformat{enumi}{#2\textup{(#1)}#3}
\labelcrefformat{enumi}{\textup{(#1)}}

\foreach \environment in 
	{definition,example,remark,notation,theorem,conjecture,proposition,lemma,%
		fact,corollary,observation,assumption,bc} 
{%
		\AtBeginEnvironment{\environment}
	{%
		
		\pushQED{\qed}
	}
	\AtEndEnvironment{\environment}{\popQED}
}

\newcommand*{\classstyle}[1]{\bm{\sf{#1}}}
\newcommand*{\setstyle}[1]{%
	\bm{#1}%
	\mathchoice{}{}{\mkern-1mu}{\mkern-1mu}%
}

\newcommand*{\problemstyle}[1]{\sf{#1}}

\newcommand*{\FO}{\classstyle{FO}}

\newcommand*{\true}{\ensuremath{\sf{true}}}
\newcommand*{\false}{\ensuremath{\sf{false}}}

\newcommand*{\PQE}{\problemstyle{PQE}}
\newcommand*{\EMPTY}{\problemstyle{EMPTY}}

\newcommand*{\RelNames}{\setstyle{R\mkern-.8mue\mkern-1.1mul}}

\newcommand*{\DB}{\setstyle{D\mkern-1muB}\scrcorr{-1mu}{1mu}}

\newcommand*{\Facts}{\setstyle{F}\scrcorr{-3mu}{0mu}}
\newcommand*{\NonEmpty}{\setstyle{N}\scrcorr{-3mu}{0mu}}

\newcommand*{\Block}{\setstyle{B}}
\newcommand*{\Blocks}{\fr{B}}

\newcommand*{\complete}[1]{ \smash{\widehat{#1}} }

\newcommand*{\SET}{\sf{set}}
\newcommand*{\SIMPLE}{\sf{simple}}

\newcommand*	{\PDB}					{\classstyle{PDB}}

\newcommand*	{\finsetPDB}			{\PDB_{<\omega}^{\SET}}

\newcommand*	{\ctblPDB}				{\PDB_{\leq\omega}}
\newcommand*	{\ctblsetPDB}			{\PDB_{\leq\omega}^{\SET}}

\newcommand*	{\TI}						{\classstyle{TI}}

\newcommand*	{\finTI}					{\TI_{<\omega}}
\newcommand*	{\finsetTI}				{\TI_{<\omega}^{\SET}}

\newcommand*	{\ctblTI}				{\TI_{\leq\omega}}
\newcommand*	{\ctblsetTI}			{\TI_{\leq\omega}^{\SET}}

\newcommand*	{\standardTI}			{\classstyle{StandardTI}}

\newcommand*	{\BID}					{\classstyle{BID}}

\newcommand*{\REL}[1]{\ensuremath{\mathsf{#1}}}				%
\newcommand*{\STR}[1]{\ensuremath{\text{\sffamily#1}}}	%
\newcommand*{\ATT}[1]{\ensuremath{\text{\sffamily#1}}}	%

\DeclareMathOperator	{\ar}		{ar}			%
\DeclareMathOperator	{\adom}	{adom}		%

\renewcommand*	{\c}	{\ensuremath{\scriptscriptstyle\complement}}	%

\newcommand*{\decode}[1]{\llbracket #1 \rrbracket_{2}}

\newcommand*{\mult}{%
	\mathord{%
	\vphantom{\#}\smash{\mathchoice
		{\text{\Large\#}}
		{\text{\Large\#}}
		{\text{\small\#}}
		{\text{\tiny\#}}
	}\mkern-2mu}%
	\scrcorr{0mu}{4mu}}

 \gdef\bcref{\hyperref[fac:bc]{Borel-Cantelli Lemma}}%

\usepackage{authblk}

\title{Independence in Infinite Probabilistic Databases}

\author[1]{Martin Grohe}
\affil[1]{\bgroup\large\url{grohe@informatik.rwth-aachen.de}\egroup}

\author[2]{Peter Lindner}
\affil[2]{\bgroup\large\url{lindner@informatik.rwth-aachen.de}\egroup}

\affil[\bgroup\empty\egroup]{\bgroup\large RWTH Aachen University\egroup}

\begin{document}

\maketitle %

\begin{abstract}
  Probabilistic databases (PDBs) model uncertainty in data.  The
  current standard is to view PDBs as finite probability spaces over
  relational database instances. Since many attributes in typical
  databases have infinite domains, such as integers, strings, or real
  numbers, it is often more natural to view PDBs as infinite
  probability spaces over database instances. In this paper, we
  lay the mathematical foundations of infinite probabilistic databases.
  Our focus then is on independence assumptions. Tuple-independent PDBs
  play a central role in theory and practice of PDBs. Here, we study
  infinite tuple-independent PDBs as well as related models such as
  infinite block-independent disjoint PDBs. While the standard model
  of PDBs focuses on a set-based semantics, we also study
  tuple-independent PDBs with a bag semantics and independence in PDBs over
  uncountable fact spaces.

  We also propose a new approach to PDBs with an open-world
  assumption, addressing issues raised by Ceylan et al.~(Proc.\ KR
  2016) and generalizing their work, which is still rooted in finite
  tuple-independent PDBs.

  Moreover, for countable PDBs we propose an approximate query
  answering algorithm.

\end{abstract}

\section{Introduction}

Probabilistic (relational) databases (PDBs)
\cite{Suciu+2011,VanDenBroeckSuciu2017} extend the relational database model by
probability distributions in order to model uncertainty. Formally, a 
probabilistic database is a probability space over database instances of some
schema. The database instances of a probabilistic database are then called
its \emph{possible worlds}.

Applications of probabilistic databases are, for example, the management of
noisy sensor data \cite{Deshpande+2004}, information extraction
\cite{FuhrRolleke1997}, data integration, and data cleaning
\cite{Andritsos+2006, DeSa+2019}. Detailed discussions of these and other
applications of probabilistic databases can be found in
\cite{AggarwalYu2009,Suciu+2011}.

\subsection{Independence Assumptions}

The most extensively studied class of PDBs is the class of (finite)
\emph{tuple-independent} probabilistic databases \cite{DalviSuciu2004}.
Therein, all \emph{facts} (that is, events of the form \enquote{tuple $t$
appears in relation $R$}) are stochastically independent.  The probabilities of
these events are called \emph{marginal fact probabilities} or \emph{marginals}.
Due to the independence, joint probabilities of facts can easily be computed by
multiplying the respective marginals. Thus, the probability space is already
uniquely determined by the marginal probabilities of the individual facts. In
order to specify a tuple-independent PDB, it therefore suffices to give the
marginal probabilities of all relevant facts.

\begin{example}[Orders at an Online Retailer]\label{ex:order}
  This example is adapted from \cite{KennedyKoch2010}. Suppose we have a
  database with a single relation $\REL{Order}$ that stores information about
  the orders made at an online retailer, and further suppose that the data is
  subject to uncertainty. We model this with a probabilistic database. For
  simplicity, we assume that the presence of different tuples is stochastically
  independent. \Cref{fig:order} depicts a database instance (on the left) drawn
  from such a tuple-independent PDB, and the basic way to represent that PDB
  (on the right). The representation consistis of a list of all possible facts,
  here $\REL{Order}(\STR{Joe},\STR{New~York},\STR{99})$,
  $\REL{Order}(\STR{Emma},\STR{Austin},\STR{70})$ and $\REL{Order}(\STR{Dave},
  \STR{Atlanta}, \STR{19})$, together with their \emph{marginal} probability,
  i.\,e. the probability of the respective fact to be present (here $0.8$,
  $1.0$ and $0.2$, respectively).

	\begin{figure}[H]
		\centering%
		\begin{varwidth}[c]{.45\textwidth}
		\begin{tcolorbox}[hbox,fancy table,tabular={ccc},title={$\REL{Order}$},before upper app={\rowcolor{verylightgray}}]
			\ATT{Customer}	& \ATT{ShipTo} 	& \ATT{Price} [\textdollar] \\
			\hline\hline%
			\STR{Emma}		& \STR{Austin}		& \STR{70}						 \\
			\hline%
			\STR{Dave}		& \STR{Atlanta}	& \STR{19}							
		\end{tcolorbox}
		\end{varwidth}
		$\quad\sim\quad$
		\begin{varwidth}[c]{.45\textwidth}\raggedright%
		\begin{tcolorbox}[hbox,fancy table,tabular={ccc | c},title={$\REL{Order}$},before upper app={\rowcolor{verylightgray}}]
			\ATT{Customer}	& \ATT{ShipTo} 	& \ATT{Price} [\textdollar]	& \ATT{Prob.} \\
			\hline\hline%
			\STR{Joe}		& \STR{New~York} 	& \STR{99} 							& \STR{0.8} \\
			\hline%
			\STR{Emma}		& \STR{Austin}		& \STR{70}							& \STR{1.0} \\
			\hline%
			\STR{Dave}		& \STR{Atlanta}	& \STR{19}							& \STR{0.2}
		\end{tcolorbox}
		\end{varwidth}
		\caption{A representation of a (tuple-independent) \emph{probabilistic
		database} of orders at an online retailer. The \enquote{$\sim$} indicates
		that the instance shown on the left is drawn at random from the PDB
		specified on the right.}\label{fig:order}
	\end{figure}

	\Cref{fig:order} is thus an encoding of a probabilistic database with $2^3$
	possible worlds, and the probability of each world is given by the product 
	of the \emph{probabilities of the present facts} times the \emph{converse
	probabilities of the non-present facts}. For this example, the instance
	shown on the left-hand side has probability $(1-0.8) \cdot 1.0 \cdot 0.2 =
	0.04$.
\end{example}

The focus of theoretical work on independence assumptions, and on
tuple-independence in particular has multiple reasons. First of all, 
probabilistic databases are non-trivial to represent: for finite probabilistic
databases, if all facts are uncertain, then the number of possible worlds is
exponential in the number of facts. Resorting to independence assumptions
sacrifices expressive power in order to facilitate representation. Second of
all, as a byproduct of the previous point, with independence assumptions, the
probability spaces that need to be discussed have a very simple structure, and
they are therefore readily accessible for theoretical investigation.

We note that the loss in expressive power is generally not as severe of an issue 
as it might seem, and can be compensated for by adding additional mechanisms on
top of the representations. For example, every finite probabilistic database
can be represented as a first order (or relational calculus) view of a
tuple-independent PDB (see \cite{Suciu+2011}).  While such a representation is
infeasible for practical matters, we can often use constraints over
tuple-independent PDBs to describe complex correlations more succinctly
\cite{JhaSuciu2012,VanDenBroeckSuciu2017}.  In practice, some systems working
with large amounts of uncertain data directly operate under the
tuple-independence assumption, for example, Knowledge Vault \cite{Dong+2014},
NELL \cite{Mitchell+2018} and DeepDive \cite{Zhang2015}.  Even beyond
tuple-independence, most existing PDB systems use independence assumptions at
some point to span large probability spaces from independent building blocks
\cite{Dalvi+2009}, cf.
\cite{GreenTannen2006,Aggarwal2009,Suciu+2011,VanDenBroeckSuciu2017}. This
includes block-independent disjoint PDBs \cite{DalviSuciu2007b} but also more
sophisticated representations \cite{Antova+2009,Widom2009}. 

The uncertainty in probabilistic databases can come in various flavors (cf.
\cite{Suciu+2011,VanDenBroeckSuciu2017}). \Cref{ex:order} exemplifies a PDB
with \emph{tuple-level uncertainty}, where there is a number
of possible tuples, but the presence of the individual tuples is subject to
uncertainty. Under \emph{attribute-level uncertainty}, there is a fixed number
of present tuples, but for each of them one (or more) attribute values are
subject to uncertainty (cf. the discussion in \cref{ex:temprec} below).
In general, both types of uncertainty can co-occur.

\subsection{From Finite To Infinite PDBs}

In the literature on probabilistic databases, whenever PDBs are formally 
introduced, they are usually defined as a probability space over
\emph{finitely} many possible worlds. Then, only finitely many instances can
have a non-zero probability and, in particular, only finitely many facts have a
non-zero marginal probability.

This assumption is clearly not natural if facts are uncertain, but involve
attributes ranging over infinite domains like the integers, real numbers, or
strings. Even if implementations put restrictions on these domains (64-bit
integers or floating point numbers, 256-character strings), the mathematical
abstraction that we use to reason about such systems is based on the ideal
infinite domains. For example, if a database records temperature measurements
(or other sensor data), then the values are reals. We may want to allow for
some noise in the data, naturally modeled by a normal distribution around the
measured values, and this already gives us a simple example of a PDB with an
infinite (even uncountable) sample space. We illustrate this situation with
our next example, and point out several problems that occur when we restrict
ourselves to finite PDBs.

\begin{example}[Temperature Measurements]\label{ex:temprec}
	Again, consider a database recording noisy temperature measurements as an
	example. Part of such a database instance is shown in \cref{fig:temprec}. 

	\begin{figure}[h]
		\centering%
		\begin{tcolorbox}[hbox,fancy table, tabular={c c c},title={$\REL{TempRec}$},before upper app={\rowcolor{verylightgray}}]
			\ATT{RoomNo}	& \ATT{Time}		& \ATT{Temp [\textdegree{}C]}\\\hline\hline
			\STR{4108}		& \STR{2021-07-01 8:00}		& \STR{21.2}\\\hline
			\STR{4108}		& \STR{2021-07-01 14:00}	& \STR{22.2}\\\hline
			\STR{4109}		& \STR{2021-07-01 8:00}		& \STR{22.1}\\\hline
			\STR{4109}		& \STR{2021-07-01 14:00}	& \STR{22.4}\\\hline
			\scriptsize$\vdots$	& \scriptsize$\vdots$		& \scriptsize$\vdots$
		\end{tcolorbox}
		\caption{A database of temperature recordings. The vertical dots are to
		indicate that the database contains many more (although finitely many)
		such entries.}\label{fig:temprec}
	\end{figure}

	Consider the following two queries against this (non-probabilistic)
	database:
	\begin{enumerate}[label=(Q\arabic*)]
		\item \emph{Has the temperature ever been between 20.2\,\textdegree{}C and
			20.5\,\textdegree{}C?} Suppose that none of our finitely many facts
			represents a recorded temperature in this range. Then the query will
			return \enquote{\STR{false}}.
		\item \emph{Is the temperature in office 4108 always lower than the one
			in office 4109?} If this is supported by our finitely many facts,
			then the query will return \enquote{\STR{true}}.
	\end{enumerate}
	Given that the underlying data is noisy, we might want to model our database
	of temperature recordings as a \emph{probabilistic} database, for
	simplicity, say, a tuple-independent one. Now reconsider (Q1) and (Q2) under
	the assumption that the data is modeled by such a PDB with finitely many
	possible facts.
	\begin{enumerate}[label=(Q\arabic*)]
		\item If the finitely many facts occurring in our PDB do not 
			contain a recorded temperature between 20.2 and 20.5\,\textdegree{}C,
			then the query (Q1) returns \enquote{\STR{false}} with probability
			$1$ (respectively \enquote{\STR{true}} with probability $0$).
		\item If for all temperature recordings occurring among the instances of,
			the values recorded in 4108 are always lower than all recordings from
			4109, then the query (Q2) returns \enquote{\STR{true}} with
			probability $1$ (respectively \enquote{\STR{false}} with probability
			$1$).
	\end{enumerate}
	However, both of these answers do not reflect what we would expect of the
	respective query on a stochastic model of the underlying uncertainty of the
	data. Instead, it seems more reasonable to assume that the events discussed
	above may have a high probability (after all, this is what is supported by
	the available data) but not exactly $1$, accounting for imprecisions in 
	measurement and potentially missing facts.

	Note that even among the technically \emph{impossible} events, there are
	distinctions in terms of plausibility. Similar to above, it is not
	reasonable to assume that the temperature in room 4108 never exceeds
	23\,\textdegree{}C \emph{with certainty}. What would be reasonable instead
	is, for example, that a temperature slightly above 
	23\,\textdegree{}C is more likely than the temperature exceeding 
	35\,\textdegree{}C. 

	In general, conclusions drawn from the probabilistic database alone, due to
	the finite (closed-world) setting, diverge from what we would infer from the
	true underlying data, even in terms of possibility and impossibility of
	events.

	One particular solution to the issues sketched here would be to replace the
	concrete temperature recordings with random variables that are distributed
	with a normal distribution whose mean is the original recording, together
	with a small variance. For example, we could replace the fact $\REL{TempRec}(
	\STR{4108}, \STR{2021-07-01~8:00}, \STR{21.2} )$ with $\REL{TempRec}(
	\STR{4108}, \STR{2021-07-01~8:00}, \mathcal N( 21.2, 0.1 ) )$,
        where
	$\mathcal N( 21.2, 0.1 )$ indicates a normally distributed random variable
	with mean $21.2$ and variance $0.1$. Note that such a PDB is basically an
	uncountably infinite block-independent disjoint PDB, with each block
	corresponding to one of the original facts.
\end{example}

Ceylan et al. \cite{Ceylan+2016} have already pointed out such issues with 
respect to the closed-world assumption of PDBs (that facts not mentioned by the
representation have probability $0$) before in the context of finite PDBs. Yet,
their approach towards tackling them is still confined to a finite setting and
therefore still exhibits the problems that we face in \cref{ex:temprec}.

\medskip

Up to date, there already exist a variety of practical probabilistic database
systems, some of which are especially designed to support infinite domains.
This includes MCDB / SimSQL \cite{Jampani+2011,Cai+2013}, PIP
\cite{KennedyKoch2010}, Orion \cite{Singh+2008a} and Trio
\cite{AgrawalWidom2009,Widom2009}. In particular, these systems can describe,
and work with PDBs as the one we propose in \cref{ex:temprec}. From the
theoretical point of view however, infinite probabilistic databases have lacked
a general framework for a long time compared to the long history of their
finite counterparts. To our knowledge, we were the first to lay such
foundations in our work \cite{GroheLindner2019,GroheLindner2020}.

Having a sound formal foundation for infinite PDBs is important,
especially because its mathematical
development proves to be far from trivial. A model of infinite
PDBs needs to be consistent with our intuitions on
the behavior of queries, and it should to include the usual finite
model as a special case. The probability spaces we obtain may be
uncountably infinite if the underlying attribute domains are.  While
in general the idea of query semantics is the same as in the finite
setting, we then also have to pay attention to whether they are still
well-defined because of measurability issues. When building database
systems, such foundational issues may seem remote, because in practice
we are always dealing with finite approximations of the infinite
space. Yet, it is desirable to have a semantics for such systems that
goes beyond a specific implementation on a specific machine, and such
a semantics will naturally refer to idealized infinite domains. Once
we have such an idealized semantics, we can argue that a specific
system adheres to it, approximately.

Our focus on tuple-independence stems from the central role of independence
assumptions in theory literature for finite PDBs. Independence assumptions make
PDBs easy to work with mathematically. Even more so, this is the case for
infinite PDBs. This makes tuple-independent and block-independent disjoint PDBs
a natural starting point for rigorous investigations of infinite PDBs.

Formally, infinite probabilistic databases are probability spaces whose sample
space consists of infinitely many database instances. Even in an infinite PDB,
each individual instance is finite; it is best here to think of database
instances as finite sets or finite bags (a.\,k.\,a. multisets) of facts. Thus,
abstractly, PDBs are probability spaces over finite sets or finite bags. In
probability theory, such probability spaces are known as point processes
\cite{DaleyVere-Jones2003,DaleyVere-Jones2008}.

\subsection{Contributions}

After carefully introducing the mathematical framework of infinite PDBs, in
this article we focus on tuple-independence in infinite PDBs and on various
generalizations of the tuple-independence assumption. The goal of our 
contribution is to broaden the understanding of independence assumptions in
infinite PDBs by identifying their abstract structure, and discussing it in 
settings of infinite set and bag PDBs.

We start by looking at countably infinite tuple-independent PDBs, which, like
finite tuple-independent PDBs, can be specified by giving all the marginal fact
probabilities $P(f)$. We show that for a countable family $\Facts$ of facts, a
tuple-independent PDB with fact probabilities $\big(P(f)\big)_{f\in\Facts}$
exists if and only if $\sum_{f\in\Facts}P(f)<\infty$ (\cref{thm:ti_series}).
We identify independent superpositions (a notion from point process theory) as
the mathematical abstraction underlying the construction of PDBs from smaller,
independent PDBs. This facilitates reasoning about infinite PDBs with
independence assumptions. We illustrate this by casting both countable 
tuple-independent and countable block-independent disjoint PDBs (PDBs made up
from independent blocks of facts such that the facts within each block are
mutually exclusive) as superposition constructions. This allows us to reobtain
our characterization for tuple-independent PDBs, and obtain an analogous
characterization of the existence of block-independent disjoint PDBs just by
using the properties of superpositions.

Rather as a side note to our main story, our discussion of countable set PDBs
is complemented by a few additional insights into the role of tuple-independent
PDBs in terms of expressiveness and query answering. We show that in countably
infinite PDBs, independence assumptions are more restrictive than their
counterpart in the finite setting: While every finite probabilistic database is
a first-order view over a tuple-independent one, this is not true in the
countably infinite case. 
We also give insights into the computability of approximate query evaluation in
countable tuple-independent PDBs: we show that query answering can be 
approximated with additive error by approximating infinite tuple-independent
PDBs with finite ones. We prove that there can be no algorithm that achieves
multiplicative approximations.

We use our countable tuple-independent PBDs in order to incorporate an
open-world assumption. Extending the ideas of \cite{Ceylan+2016}, we construct
potentially infinite open-world completions of a finite or infinite PDB. The
key-requirement is that the probability measure is faithfully extended: in an
open-world completion of a PDB, the probability measure should coincide with
the original one, when conditioned over the sample space of the original PDB. 

\smallskip

Up to this point, we have treated PDBs with a set semantics, but we
can extend our definitions and results to PDBs with a bag semantics. %
With the machinery of superpositions, it is relatively easy to prove
a general existence result for tuple-independent PDBs with a bag semantics,
where we simply combine the distributions of individual fact multiplicities
using a superposition. %
When it comes to a treatment of PDBs with an uncountable sample space,
a bag semantics turns out to be easier to handle. Even if one is only
interested in PDBs with a set semantics, bag semantics is a usueful
(and to some extent necessary) intermediate step.

Note that a generalization PDBs to uncountable spaces is important,
because in many applications we have real-valued attributes. The generalization
of tuple-independence is not completely straightforward, because typically in
an uncountable setting the individual fact probabilities will be zero. 
When dealing with uncountable PDBs, we build on the notion of
\emph{standard PDBs} introduced in \cite{GroheLindner2020,GroheLindner2022}.
Our treatment heavily draws from the mathematical theory of finite
point processes \cite{DaleyVere-Jones2003} and the notion of completely random measures
\cite{Kingman1967}.

\bigskip 

A conference version of this article, which contains the basic results for
countable infinite PDBs, has been presented at the 38th ACM SIGMOD-SIGACT-SIGAI
Symposium on Principles of Database Systems (PODS 2019)
\cite{GroheLindner2019}. However, the uniform construction principle of PDBs
from independent building blocks based on superpositions is new here.  In a
subsequent paper \cite{GroheLindner2020,GroheLindner2022}, we developed a
generic framework for uncountable PDBs, the so-called standard PDBs, and this
framework allowed us to extend our theory of independent PDBs to the
uncountable case in this article.  Building on our work, Carmeli et
al.~\cite{Carmeli+2021} studied the power of tuple-independent PDBs as a
representation system for countably infinite PDBs.

\subsection{Related Work}

The foundation of our work is the extensive literature on models for finite
probabilistic databases
\cite{GreenTannen2006,Aggarwal2009,Suciu+2011,VanDenBroeckSuciu2017}. Dalvi et
al. identified three \emph{facets} of research in probabilistic data management
\cite{Dalvi+2009}: \emph{semantics and representation}; \emph{query
evaluation}; and \emph{user interface} principles. Among these, our 
contribution mainly addresses semantics of probabilistic data in an infinite 
setting, with emphasis on independence assumptions. Concrete representation
systems \cite{GreenTannen2006,Suciu+2011}, query evaluation and user interfaces 
are not the focus of our work, but only occasionally touched upon.
Note that discussions of independence assumptions are abundant
in PDB literature (cf.
\cite{GreenTannen2006,Suciu+2011,VanDenBroeckSuciu2017}). In the finite
setting, such assumptions are exploited for concise representations of large
probability spaces over database instances.

\emph{Incomplete databases} \cite{ImielinskiLipski1984,vanderMeyden1998} are a 
non-probabilistic model of uncertain databases. As opposed to probabilistic
databases, models of incomplete databases typically do not assume finite
domains \cite{ImielinskiLipski1984,vanderMeyden1998,GreenTannen2006}.
Essentially, a PDB augments an incomplete database with probabilities and thus,
the problem of their representation is closely related to that of PDBs
\cite{GreenTannen2006}. Most of the PDB literature assumes this model, referred 
to as the \emph{possible worlds semantics} \cite{Dalvi+2009,Suciu+2011}. 

There exist system-oriented approaches to PDBs that are able to handle
continuous data, such as MCDB / SimSQL \cite{Jampani+2011,Cai+2013}, PIP
\cite{KennedyKoch2010}, Orion \cite{Singh+2008a} and Trio
\cite{AgrawalWidom2009,Widom2009}. Earlier, Dalvi et al. \cite{Dalvi+2009}
noted the insufficient understanding of models for uncountable PDBs in terms of
possible worlds semantics. Despite being over 10 years old, this statement is
for the most part still valid today. The data model of Orion 
\cite{Singh+2008a,Singh+2008b} is worth mentioning because it explicitly draws 
a connection from uncountable PDBs to the notion of possible worlds but only
allows a bounded number of tuples per PDB. Measure theoretic approaches to PDBs
in terms of possible worlds semantics are scarce. %

Also for \emph{probabilistic XML} \cite{Abiteboul+2009,KimelfeldSenellart2013}, 
an infinite model covering continuous distributions has been introduced 
\cite{Abiteboul+2011}. Yet, this approach leaves the document structure finite.
On the contrary, in \cite{Benedikt+2010}, the authors propose an unbounded
model of probabilistic XML that does not support continuously distributed data.

Query answering in PDBs is closely related to the problem of \emph{weighted 
model counting (WMC)} \cite{VanDenBroeckSuciu2017}, that is, to counting the
models of a logical sentence in a weighted way. In recent work, the WMC problem
was extended to infinite domains as well \cite{Belle2017}. The idea of 
completions of probabilistic databases is introduced in OpenPDBs
\cite{Ceylan+2016,FriedmanVanDenBroeck2019}, as a means to overcome problems 
arising from the closed-world assumption in probabilistic databases
\cite{Reiter1981,ZimanyiPirotte1997} as lined out before. The work
\cite{Stoyanovich+2011} essentially describes a model of block-independent
disjoint completions of a given \emph{incomplete} database where the
probabilities of the missing facts are inferred from the existing ones.
In \cite{Borgwardt+2017,Borgwardt+2018}, ontologies are used to improve query
results on PDBs.

Some related fields of research (in particular, artificial intelligence and
machine learning \cite{DeRaedt+2016,Belle2020} and probabilistic programming
\cite{Gordon+2014}) have developed approaches towards infinite probabilistic
data models before as modeling languages or systems such as BLOG
\cite{Milch+2005} or Markov Logic Networks
\cite{RichardsonDomingos2006,SinglaDomingos2007,WangDomingos2008}, as abstract
data types \cite{Faradjian+2002}, and various probabilistic programming
languages as, for example, ProbLog \cite{DeRaedt+2007,Gutmann+2011} and others
\cite{Goodman+2008,Pfeffer2009,Tolpin+2015}. On the side of programming
languages, Probabilistic Programming Datalog \cite{Barany+2017,Grohe+2020} has
direct ties to PDBs. Programming with point processes for practical stochastic
models has been recently cast into a monadic framework with category-theoretic
foundations \cite{DashStaton2021}. From a very abstract point of view the basic
model presented there is the model from \cite{GroheLindner2020} that we use for
uncountable PDBs in \cref{s:beyond}. A construction of probability spaces over
sets of facts from a countably infinite product space of facts already appeared
in the probabilistic programming community as \emph{distribution semantics}
\cite{Sato1995}.

Infinite relational structures also play a role in the discussion of limit
probabilities in asymptotic combinatorics \cite{Bollobas2001,Spencer2001}. For
example, the classical Erd\oldH{o}s-Rényi model $\G(n,p)$ is essentially 
\enquote{tuple-independent}: it describes a probability distribution over 
$n$-vertex graphs where every possible edge is drawn independently with 
probability $p$. Usually, the focus lies on studying properties of this model
as $n \to \infty$. In such a model, properties of large graphs (or databases)
dominate the observed behavior. This contrasts the infinite tuple-independence
model we discuss here, that is dominated by instances in the vicinity of its
expected instance size (which is always finite for tuple-independent PDBs).
There exists work studying limit probabilities in PDBs \cite{Dalvi+2004} and
incomplete databases \cite{Libkin2018,Console+2020}.

\subsection{Organization of the Article}

After giving the necessary preliminaries in \cref{sec:prel} (with additional 
background material being contained in the appendix), we introduce
probabilistic databases (both finite and infinite) in
Section~\ref{s:relpdb}. Section~\ref{s:countable} contains our
results on countable tuple-independent PDBs with a set semantics and
generalizations such as block-independent disjoint PDBs. Specifically,
in Subsection~\ref{ss:sp} we introduce countable superpositions.
Section~\ref{s:bag_instances} is devoted to countable tuple-independent PDBs
with a bag semantics. Finally, in Section~\ref{s:beyond} we consider
tuple-independence in the general setting of (potentially uncountable) PDBs. We
conclude with a few remarks and open questions in Section~\ref{s:conclusion}.

\section{Notation and Mathematical Background}
\label{sec:prel}

Throughout this paper, $\NN$ denotes the set of non-negative integers and
$\NN_{>0}$ denotes the set of positive integers. The set of real numbers is
denoted by $\RR$. We denote open, half-open, and closed intervals of reals by
$(a,b)$, $[a,b)$, $(a,b]$, and $[a,b]$.

We call a set or collection \emph{countable} if it is finite or countably
infinite. We denote the powerset of a set $S$ (that is, the set of subsets of
$S$) by $\powerset(S)$. The \emph{finitary powerset} of $S$ (that is, the set
of finite subsets of $S$) is denoted by $\powerset_{\fin}(S)$.

\subsection{Bags and Sets}

A \emph{bag} (or \emph{multiset}) is a pair $B = (S_B,\mult_B)$ where $S_B$ is 
a set and $\mult_B$ is a function $\mult_B \from S_B \to \NN \cup \set{ \infty
}$. We call $B$ a \emph{bag over $S_B$} and $\mult_B$ its \emph{multiplicity
function}. For $S \subseteq S_B$, we let $\mult_B( S ) \coloneqq \sum_{ s \in S
} \mult_B(s)$. The cardinality $\size{B}$ of a bag $B = (S_B, \mult_B )$ is the
sum of all its multiplicities, that is, $\size{B} = \mult_B(S_B) = \sum_{ s \in
S_B } \mult_B(s) \in \NN \cup \set{ \infty }$. A bag $B$ is called
\emph{finite} if $\size{B} < \infty$. The sets of all bags over some set $S$ is
denoted by $\powerbag(S)$.  The set of all finite bags over $S$ is denoted by
$\powerbag_{\fin}( S )$. We identify any bag $B = ( S_B, \mult_B )$ where
$\mult_B$ is $\set{ 0, 1 }$-valued with the set $\set{ s \in S_B \with \mult_B(
s ) = 1 }$. 

Let $B_1 = (S_1, \mult_1)$ and $B_2 = (S_2, \mult_2)$ be bags. The
\emph{additive union} of $B_1$ and $B_2$ is the bag
\begin{equation*}
	B_1 \uplus B_2 \coloneqq (S_1 \cup S_2, \mult_1 + \mult_2)
\end{equation*}
with $( \mult_1 + \mult_2 )( s ) = \mult_1( s ) + \mult_2( s )$ and the
convention that $\mult_i( s ) = 0$ if $s \notin S_i$ for $i \in \set{ 1, 2 }$.
The additive union of bags is associative and commutative. Moreover, for every
bag $B$ it holds that $B \uplus \emptyset = B$ where $\emptyset$ is the empty
bag. We write $\biguplus_{ i=1 }^n B_i$ for $B_1 \uplus \dots \uplus B_n$ for
all $n \in \NN_{ >0 }$ and bags $B_1,\dots,B_n$. Any additive union of finitely
many finite bags is finite.

\subsection{Infinite Sums and Products}\label{sec:infprod}

In comparison with finite probabilistic databases, the discussion of infinite
probabilistic databases brings with it some new analytic challenges, as even in
the most simple examples, we need to take sums and products over infinite index
sets. We will therefore frequently encounter series $\sum_{ i = 0 }^{ \infty }
a_i$ and infinite products $\prod_{ i = 0 }^{ \infty } a_i$ with terms $a_i \in
[0,1]$.  All series as above attain well-defined values in $[0, \infty]$. If
the value is finite, the series is called \emph{convergent} (and
\emph{divergent} otherwise). All products as above attain well-defined values
in the interval $[0,1]$. Moreover, all series and products in this paper have
the property that their value is independent of the order of terms.
\Cref{app:infsumprod} contains some more formal background on infinite series
and products.

\subsection{Probability and Measure Theory}

Here in the preliminaries, we only cover the bare minimum background from
probability theory. We refer to \cref{app:prob} for the formal definitions and
statements, and pointers to textbooks.

A \emph{measure space} is a tuple $( \Omega, \AAA, \mu )$ where $\Omega$ is
some non-empty set, $\AAA$ is a $\sigma$-algebra\footnote{A $\sigma$-algebra on
$\Omega$ is a family of subsets of $\Omega$ that contains $\Omega$ and is
closed under countable unions, and under taking complements.} on $\Omega$, and
$\mu \from \AAA \to [0,\infty]$ is a function with the property that $\mu(
\emptyset ) = 0$ and $\mu\big( \bigcup_{ i = 0 }^{ \infty } \AAAA_i \big) =
\sum_{ i = 0 }^{ \infty } \mu( \AAAA_i )$ whenever $\AAAA_1, \AAAA_2, \dots \in
\AAA$ are pairwise disjoint. This latter property of $\mu$ is called
\emph{$\sigma$-additivity}. The elements of $\AAA$ are called \emph{measurable
sets}. Measure spaces with $\mu( \Omega ) = 1$ are called \emph{probability
spaces}. In this case, $\Omega$ is called the \emph{sample space}, the sets in
$\AAA$ are called \emph{events}, and $\mu$ is called a \emph{probability
measure}. Probability measures are usually denoted by $P$.

Let $(\Omega, \AAA, P)$ be a probability space. We write
\begin{equation*}
	\Pr_{ X \sim ( \Omega, \AAA, P ) } \event{ X \in \AAAA } \coloneqq
	P( \AAAA )
\end{equation*}
for the probability of the event $\AAAA \in \AAA$. If $\Phi \from \Omega \to
\set{ \true, \false }$ is a measurable Boolean property (for example, given by
a sentence in some logic, if $\Omega$ contains relational structures), then we
write
\begin{equation*}
	\Pr_{ X \sim ( \Omega, \AAA, P ) } \event{ \Phi( X ) }
	\coloneqq \Pr_{ X \sim ( \Omega, \AAA, P ) } \event{ X \text{ has property }
	\Phi }
	= P\big( \set{ \omega \in \Omega \with \Phi( \omega ) = \true } \big)\text.
\end{equation*}

\phantomsection\label{fac:bc}
Events $\AAAA_0, \AAAA_1, \dots$ in a probability space $(\Omega,\AAA,P)$ are
called \emph{independent}, if the joint probability of any finite subset of
these events is the product of their individual probabilities, that is, if 
$P\big( \bigcap_{i\in I} \big) = \prod_{i\in I} P( A_i )$ for all finite
$I\subseteq \NN$. A family of events is independent if and
only if their complements are. An important statement we need is the \bcref{}
\cite[Theorem~2.7]{Klenke2014}, which states that if $\AAAA_1, \AAAA_2, \dots
\in \AAA$ are events in a probability space $( \Omega, \AAA, P )$, then 
\[
	\sum_{ i = 1 }^{ \infty } P( \AAAA_i ) < \infty
		\quad
		\Rightarrow
		\quad
	P\bigg( \bigcap_{ i = 1 }^{ \infty } \bigcup_{ j = i }^{ \infty } \AAAA_j \bigg)
		= \Pr_{ X \sim ( \Omega, \AAA, P ) } \event{ X \in \AAAA_i \text{ for inf.\ many }i } = 0
	\text.
\]
If the $\AAAA_i$ are additionally pairwise independent, then 
\[
	\sum_{ i = 1 }^{ \infty } P( \AAAA_i ) = \infty
		\quad
		\Rightarrow
		\quad
	P\bigg( \bigcap_{ i = 1 }^{ \infty } \bigcup_{ j = i }^{ \infty } \AAAA_i \bigg)
		= \Pr_{ X \sim ( \Omega, \AAA, P ) } \event{ X \in \AAAA_i \text{ for inf.\ many }i } = 1
	\text.
\]

Given probability spaces $( \Omega_i, \AAA_i, P )$, $i = 1,2, \dots$, there
exists a unique probability measure $P$ on $\prod_{ i = 1 }^{ \infty }
\Omega_i$ (equipped with a suitable $\sigma$-algebra, that is, the product
$\sigma$-algebra) such that the events $\big(\Omega_1 \times \dots \times
\Omega_{i-1} \times \AAAA_i \times \Omega_{i+1} \times \dots\big)_{ i \in
\NN_{>0}}$ are independent, and have probability $P_i( \AAAA_i )$ where
$\AAAA_i \in \AAA_i$.

\subsection{Relational Databases}

In this article, we consider the unnamed perspective on the relational model
\cite{Abiteboul+1995}. Let $\RelNames$ be some countably infinite set. The
elements of $\RelNames$ are called \emph{relation symbols}. We fix a function
$\ar \from \RelNames \to \NN$. Then, for all $R \in \RelNames$, $\ar(R)$ is
called the \emph{arity} of $R$.

A \emph{database schema} $\tau$ is a finite set of relation symbols. Let $\UU
\neq \emptyset$ be some set of arbitrary (positive) cardinality, called the
\emph{universe} or \emph{domain}. Then a \emph{$ ( \tau, \UU ) $-fact} is an
expression of the shape $R ( u )$ where $u \in \UU^{ \ar(R) }$. We let
$\Facts[\tau,\UU]$ denote the set of $( \tau, \UU )$-facts.

A \emph{database instance} $D$ of schema $\tau$ over $\UU$ (or, 
\emph{$( \tau, \UU )$-instance}) is a finite bag of $( \tau, \UU )$-facts. We 
let $\DB[\tau, \UU]$ denote the set of all $( \tau, \UU )$-instances. If $D 
\in \DB[\tau,\UU]$ has only $\set{0,1}$-valued fact multiplicities, we call $D$ 
a \emph{set instance}. Otherwise, $D$ is called a \emph{(proper) bag instance}.
We let $\DB^{\SET}[\tau,\UU]$ denote the set of all set instances of schema
$\tau$ over $\UU$. The active domain $\adom(D)$ of a database instance $D$ is
the restriction of $\UU$ to the elements appearing in $D$.

\section{Probabilistic Databases}
\label{s:relpdb}

The following definition of possibly infinite, even uncountable probabilistic
databases is the straightforward generalization of the traditional definition of
PDBs as finite probability spaces. Besides the possibility of infinite
probability spaces, it also allows for bag instances and treats set instances as
a special case.

\begin{definition}
	\label{def:pdb}
	A \emph{probabilistic database (PDB)} of schema $\tau$ over $\UU$ (or, 
	\emph{$(\tau,\UU)$-PDB}) is a probability space $\D = (\DB,\DDD,P)$ with
	$\DB \subseteq \DB[\tau,\UU]$. The PDB $\D$ is called
	\begin{itemize}
		\item a \emph{set PDB}, if $\DB \subseteq \DB^{\SET}[\tau,\UU]$;
		\item a \emph{simple PDB}, if it holds that $\Pr_{ D \sim \D } 
			\event{ D \in \DB^{\SET}[\tau,\UU] } = 1$.
	\end{itemize}
	A PDB $\D = (\DB,\DDD,P)$ is called \emph{finite},  
	\emph{countable}, \emph{countably infinite} or \emph{uncountable} if 
	$\DB$ has the corresponding cardinality. The class of all PDBs is denoted by
	$\PDB$. The classes of set PDBs and simple PDBs are denoted by $\PDB^{\SET}$ 
	and $\PDB^{\SIMPLE}$, respectively. We use the subscripts \enquote{$ < 
	\omega$} and \enquote{$ \leq \omega $} to refer to respective subclasses of
	finite and countably infinite PDBs.
\end{definition}

If $\D = (\DB,\DDD,P)$ is countable, then we always assume that $\DDD =
\powerset(\DB)$ is the powerset $\sigma$-algebra on $\DB$ and write $(\DB,P)$
instead of $(\DB,\powerset(\DB),P)$. In principle, we can always restrict the
sample spaces to an arbitrary subspace that carries all the probability mass.
In particular, for countable PDBs, there is no need for any distinction between
set and simple PDBs. In the uncountable, where measurability is not trivial
anymore, it is mathematically more convenient to keep the general sample
spaces. There, although happening with probability $0$, an outcome drawn from
a simple PDB may contain duplicates.

We note that for most applications involving uncountable PDBs, additional
conditions on the $\sigma$-algebras are needed. For example, we would want that
queries of typical database query languages have a well-defined semantics. In
order for this to be the case, the queries need to be measurable functions
between PDBs, requiring the measurability of various kinds of events. While the
discussion of such is generally beyond the scope of this paper (and discussed
in detail in \cite{GroheLindner2020}), these issues will be revisited in
\cref{s:beyond} when we explicitly discuss uncountable PDBs.

From the point of view of probability theory, PDBs are special finite point
processes\footnote{Point process theory is also from where we borrow the term 
\enquote{simple} for PDBs whose outcomes are set instances with probability
$1$.}. A point process is a random collection of points (usually allowing
duplicates) in some measurable space \cite{DaleyVere-Jones2003}. A finite point
process is a random \emph{finite} collection of points. Point processes can be
equivalently described as random integer-valued measures \cite{Kallenberg2017}.

\subsection{Marginal Probabilities}
\label{ss:marginal-probabilities}
Let $\D = (\DB, \DDD, P)$ be a PDB. The set of facts appearing in $\DB$ is
denoted by $\Facts(\D)$. We say that $\D$ is a PDB \emph{with fact set} (or,
\emph{over}) $\Facts(\D)$. For every fact $f \in \Facts(\D)$ with the property
that $\set{ D \in \DB \with \mult_D( f ) > 0 } \in \DDD$, the \emph{marginal
probability} of $f$ in $\D$ is given by
\begin{equation*}
	P( f )
	\coloneqq \Pr_{D \sim \D} \event{ f \in D }
	= P \big( \set{ D \in \DB \with \mult_D( f ) > 0 } \big)
	\text.
\end{equation*}
Potentially, we also want to include facts of marginal probability $0$ in our
sample spaces. As soon as we move to non-discrete distributions on uncountable 
spaces, this is necessary anyway. In fact, for uncountable PDBs it may happen
naturally that \emph{all} facts have marginal probability $0$.

\begin{example}\label{ex:single_fact_pdf}
	\Cref{fig:temprecprob} depicts a representation of an uncountable set PDB 
	for the temperature records example (\cref{ex:temprec}). Formally, the fact 
	space is
	\[ 
		\Facts[\tau, \UU] = \set{ \REL{TempRec}(r,t,c) \with r,t \in 
	\Sigma^*\text{ and }c \in \RR}\text,
	\]
	where $\Sigma$ is some alphabet.
	\begin{figure}[h]
		\centering%
		\begin{tcolorbox}[hbox,fancy table, tabular={c c c},title={$\REL{TempRec}$},before upper app={\rowcolor{verylightgray}}]
			\ATT{RoomNo}	& \ATT{Time}		& \ATT{Temp [\textdegree{}C]}	\\\hline\hline
			\STR{4108}		& \STR{2021-07-01 8:00}		& $\mathcal N(21.2,0.1)$\\\hline
			\STR{4108}		& \STR{2021-07-01 14:00}	& $\mathcal N(21.2,0.1)$\\\hline
			\STR{4109}		& \STR{2021-07-01 8:00}		& $\mathcal N(22.1,0.1)$\\\hline
			\STR{4109}		& \STR{2021-07-01 14:00}	& $\mathcal N(22.4,0.1)$
		\end{tcolorbox}
		\caption{Representation of an (uncountably) infinite
		PDB.}\label{fig:temprecprob}
	\end{figure}

	The probability measure of this PDB is given by the joint distribution of
	the four normally distributed random variables explicitly listed in
	\cref{fig:temprecprob}. Then, the marginal probability of any fact, say, for
	example, of the fact 
	\[
		\REL{TempRec}( \STR{4108}, \STR{2021-07-01 8:00}, \STR{21.2} )
	\]
	is the probability of drawing a particular temperature value (here $21.2$)
	from a normal distribution, hence, $0$. Yet, with probability $1$, there are
	exactly $3$ facts in a randomly drawn instance.
\end{example}

\Cref{ex:single_fact_pdf} is closely related to PDB models based on attribute
level uncertainty \cite{Singh+2008b,AgrawalWidom2009} that allow for
continuously distributed attributes.

\subsection{Expected Instance Size}
\label{sss:expected_instance_size}

Let $\D = ( \DB, \DDD, P )$ be a PDB. The \emph{instance size function} of $\D$
is the function $\size{ \? } \from \DB \to \NN$ that maps every instance $D \in
\DB$ to its cardinality $\size{ D }$. If $\size{ \? }$ is measurable, that is,
if $\set{ D \in \DB \with \size{ D } = k } \in \DDD$ for all $k \in \NN$, then
$\size{ \? }$ is a random variable. Note that whether this is the case depends 
on the $\sigma$-algebra $\DDD$ that $\DB$ is equipped with in $\D$. If $\D =
( \DB, P )$ is a countable PDB (equipped with the powerset $\sigma$-algebra),
then $\size{ \? }$ is always a random variable.

If $\size{ \? }$ is a random variable, then its expectation $\Expectation_{ \D
}\big( \size{ \? } \big)$ in $\D$ is given as
\begin{equation}\label{eq:expect_n}
	\Expectation_{ \D }\big( \size{ \? } \big)
	= \sum_{ n = 0 }^{ \infty } n \cdot 
	\Pr_{ D \sim \D } \event[\big]{ \size{ D } = n }
	= \sum_{ n = 0 }^{ \infty } n \cdot
		P \event[\big]{ D \in \DB \with \size{ D } = n }
	\text.
\end{equation}
Note that the above also holds if $\D$ is uncountable. (All we do is partition
the range of the discrete size random variable.)

Note that for all $D \in \DB$, it holds that $\size{ D } = \mult_{ D }\big(
\Facts( \D ) \big) = \sum_{ f \in \Facts( \D ) } \mult_{D}( f )$, where the
latter sum has countably (indeed, finitely) many non-zero terms. If $\Facts( \D
)$ is countable, then\footnote{The expectation of a sum of countably many non-negative random variables is equal to the sum of the expectations, see \cite[Theorem 5.3(vi)]{Klenke2014}.} it holds that
\begin{equation}\label{eq:ctbl_expectation_measure}
	\Expectation_{ \D }\big( \size{ \? } \big)
	= \Expectation_{ \D }\Big( \mult_{ ( \? ) }\big( \Facts(\D) \big) \Big)
	= \Expectation_{ \D }\bigg( \sum_{ f \in \Facts( \D ) } \mult_{ (\?) }( f ) \bigg)
	= \sum_{ f \in \Facts( \D ) } \Expectation_{ \D }\big( \mult_{ (\?) }( f )\big)\text.
\end{equation}
Moreover, for set PDBs, $\mult_{ (\?) }( f )$ is exactly the indicator random
variable of the event $\event{ f \in D }$, $D \sim \D$. Thus, for countable
set PDBs, \labelcref{eq:ctbl_expectation_measure} entails that
\[
	\Expectation_{ \D }\big( \size{ \? } \big)
	= \sum_{ f \in \Facts( \D ) }
		\Expectation_{ \D }\big( \mult_{ (\?) }( f ) \big)
	= \sum_{ f \in \Facts( \D ) }
		P( f )
\]
If $\Expectation_{\D}\big( \size{ \? } \big) = m$ with $m \in \RR_{ \geq 0 }
\cup \set{ \infty }$, we say that $\D$ is \emph{of expected size $m$}. Although
all database instances are finite, it is very easy to construct PDBs of
infinite expected size. Intuitively, in such PDBs the sizes of instances grow
too fast to be compensated by their probabilities.

\begin{example}\label{ex:infiniteexpectation}
	Let $P$ be any probability distribution on $\NN$ with $P( n ) > 0$ for
	infinitely many $n \in \NN$. Suppose $\UU$ is some countably infinite
	universe and $\tau = \set{ R }$ with $\ar( R ) = 1$.  Let $D_0, D_1, D_2,
	\dots$ be any sequence of pairwise distinct database instances over $\tau$
	and $\UU$ such that $\size{ D_n } > \frac1{ P( n ) }$ and consider the PDB
	$\D$ with $\DB = \set{ D_0, D_1, D_2, \dots }$, $\DDD = \powerset( \DB )$
	and $\Pr_{ D \sim \D }\event{ D = D_n } = P( n )$. Then
	\[
		\Expectation_{ \D }\big( \size{ \? } \big) 
		= \sum_{ n = 0 }^{ \infty } \: \size{ D_n } \cdot P( n )
		> \sum_{ n = 0 }^{ \infty } 1 = \infty \text.
	\]
	As a concrete example of this construction, take $P( n ) = \frac6{ \pi^2 }{
	n^2 }$ for all $n \in \NN$. Note that this is indeed a probability
	distribution as $\sum_{ n = 0 }^{ \infty } \frac1{n^2} = \frac{\pi^2}6$. Let
	$\UU = \NN$ and consider the instances $D_n = \set{ R(1), \dots, R(2^n) }$,
	$n \in \NN$. Then $\Expectation_{ \D }\big( \size{ \? } \big) = \sum_{ n = 0
	}^{ \infty } 2^n \cdot \frac6{\pi^2n^2} = \infty$.
\end{example}

\subsection{Superpositions}
\label{sss:fin_sp}

(Independent) superposition is a standard operation of point processes 
\cite{DaleyVere-Jones2003,DaleyVere-Jones2008,LastPenrose2017} that we apply to 
PDBs. They provide a useful abstract tool that we can use to model how
independence assumptions are cast into probabilistic databases from independent
building blocks. In this section, for getting started, we give the basic idea
for a superposition of two countable PDBs.

Suppose $\D_1 = ( \DB_1, P_1 )$ and $\D_2 = ( \DB_2, P_2 )$ are countable PDBs
over $\tau$ and $\UU$ (equipped with the powerset $\sigma$-algebras). The
additive union $\uplus$ of bags can be lifted to a function between PDBs in a
straightforward way via $\D_1 \uplus \D_2 \coloneqq ( \DB[ \tau, \UU ], P )$ 
where $P$ is defined by
\begin{equation*}
	P\big( \set{ D } \big) \coloneqq 
	\mkern6mu\smashoperator{\sum_{ \substack{ 
		D_1 \in \DB_1\text,\\
		D_2 \in \DB_2\text,\\
		D_1 \uplus D_2 = D
	}}}\mkern6mu P_1 \big( \set{ D_1 } \big) \cdot P_2 \big( \set{ D_2 } \big)
	\text.
\end{equation*}
This indeed defines a probability measure on $\DB[ \tau, \UU ]$, as
\begin{align*}
	P\big( \DB[ \tau, \UU ] \big)
	= \mkern12mu\smashoperator{\sum_{ D \in \DB[ \tau, \UU ] }}\mkern16mu 
		P\big( \set{ D } \big)
	&= \mkern6mu\smashoperator{\sum_{ D_1 \in \DB_1 }}\mkern24mu
		\smashoperator{\sum_{ D_2 \in \DB_2 }}\mkern8mu
		P_1\big( \set{ D_1 } \big) \cdot P_2\big( \set{ D_2 } \big)\\
	&= \mkern6mu\smashoperator{\sum_{ D_1 \in \DB_1 }}\mkern8mu 
		P_1\big( \set{ D_1 } \big) \cdot
		\mkern6mu\smashoperator{\sum_{ D_2 \in \DB_2 }}\mkern8mu 
		P_2\big( \set{ D_2 } \big)
	= \mkern6mu\smashoperator{\sum_{ D_1 \in \DB_1 }}\mkern8mu 
		P_1\big( \set{ D_1 } \big)
	= 1\text.
\end{align*}
Note that this relied on both $\D_1$ and $\D_2$ being discrete probability 
spaces. For uncountable PDBs, defining $P$ on singletons would not suffice
(in fact, it could be that all singletons in both $\D_1$ and $\D_2$ have
probability $0$). The idea to extend this is to derive $P$ from the product
measure of $P_1$ and $P_2$. 

The PDB $\D_1 \uplus \D_2$ is called the \emph{(independent) superposition} of
$\D_1$ and $\D_2$. The superposition of PDBs is associative. We write
$\biguplus_{ i = 1 }^{ n } \D_i$ for $\D_1 \uplus \dots \uplus \D_n$.

\begin{example}
	Consider the two PDBs $\D_1$ and $\D_2$ depicted below.
	\begin{figure}[H]
		\centering%
		\begin{tabular}{c ccc}
			\toprule
			$D$
				& $\bag{ f }$
				& $\bag{ f, f }$\\
			\midrule
			$P_1\big(\set{ D }\big)$
				& $\tfrac34$
				& $\tfrac14$\\
			\bottomrule
		\end{tabular}
		\qquad
		\begin{tabular}{c ccc}	
			\toprule
			$D$
				& $\emptyset$
				& $\bag{ f }$
				& $\bag{ f' }$\\
			\midrule
			$P_2\big(\set{D}\big)$
				& $\tfrac38$
				& $\tfrac38$
				& $\tfrac28$\\
			\bottomrule
		\end{tabular}
	\end{figure}
	The following is the independent superposition of $\D_1$ and $\D_2$:
	\begin{figure}[H]
		\centering%
		\begin{tabular}{c ccccc}
			\toprule
			$D$
				& $\bag{ f }$
				& $\bag{ f, f }$
				& $\bag{ f, f' }$
				& $\bag{ f, f, f }$
				& $\bag{ f, f, f' }$\\
			\midrule
			$P\big(\set{D}\big)$
				& $\tfrac9{32}$
				& $\tfrac{12}{32}$
				& $\tfrac6{32}$
				& $\tfrac3{32}$
				& $\tfrac2{32}$\\
			\bottomrule
		\end{tabular}%
	\end{figure}
	Obviously, this is again a PDB.
\end{example}

Point process theory naturally considers also superpositions of countably many
point processes \cite{LastPenrose2017} and they will naturally appear
throughout this article. However, in general, the superposition
\enquote{$\biguplus_{i = 1}^\infty \D_i$} of countably infinitely many PDBs
$\D_i$ may fail to be a PDB itself: In the \enquote{result}, the multiplicity
of a single fact could be infinite or there could be infinitely many different
facts in a single \enquote{instance} with positive probability. This is
investigated in detail in \cref{ss:sp}.

\section{Countable Probabilistic Databases}
\label{s:countable}

The main subject of this section is the investigation of generalizations of
independence assumptions as they are known from finite PDBs \cite{Suciu+2011} 
to countably infinite ones.

Throughout the whole section, we make the following assumptions:
\begin{enumerate}[label=(\Roman*)]
\item\label[assumption]{ass:ctbl1} 
  We only consider a fixed database schema $\tau$ and a fixed universe 
  $\UU$ of countable size.
\item\label[assumption]{ass:ctbl2}
  Whenever we consider sets $\Facts$ of facts, then $\Facts \subseteq
  \Facts[\tau, \UU]$. Consequentially, \enquote{facts} always means 
  $(\tau, \UU)$-facts.
\end{enumerate}
Moreover, in Subsections~\ref{ss:ti}--\ref{ss:approx_eval} we make the
following assumption.
\begin{enumerate}[label=(\Roman*)]
  \addtocounter{enumi}{2}
\item\label[assumption]{ass:ctbl3} 
		If not explicitly stated otherwise, all PDBs that occur in this 
		section have sample space $\powerset_{ \fin }\big(
      \Facts( \D ) \big)$ and are equipped with the powerset $\sigma$-algebra. 
\end{enumerate}
That is, unless explicitly
stated otherwise, we use a \emph{set} semantics.

\subsection{Countable Tuple-Independence}
\label{ss:ti}

With this subsection, we begin our investigation of independence assumptions
in infinite PDBs by extending the well-known notion of tuple-independent PDBs
towards PDBs over infinite sets of facts. We discuss the circumstances under
which such PDBs exist and, if so, how to construct them by generalizing the
finite construction in the natural way.

\begin{definition}[Countable $\TI$-PDBs]\label{def:ti_pdb}
	A PDB $\D = ( \DB, P )$ is called \emph{tuple-independent} (or, a
	\emph{$\TI$-PDB}) if the events $\set{ D \in \DB \with f \in D }$, for
	$f\in\Facts(\D)$, are independent, that is, if for all $k = 1, 2, \dots$ and
	all pairwise different $f_1, \dots, f_k \in \Facts( \D )$ it holds that
	\begin{equation}\label{eq:independence_across_facts}
		\Pr_{ D \sim \D } \event[\big]{ f_1, \dots, f_k \in D }
		= P( f_1 ) \cdot \dotso \cdot P( f_k )
		\text.\qedhere
	\end{equation}
\end{definition}

The property from \cref{eq:independence_across_facts} is referred to as
\emph{independence across} (or, \emph{of}) \emph{facts}. We denote the subclass
of countable tuple-independent set PDBs of $\ctblsetPDB$ by $\ctblsetTI$. The
class of \emph{finite} tuple-independent set PDBs is denoted by $\finsetTI$.

In the finite setting, a $\TI$-PDB $\D$ may be specified by providing its facts
$\Facts \coloneqq \Facts( \D )$ together with their marginal probabilities
$\big( P( f ) \big)_{ f \in \Facts }$, and any such pair $( \Facts, P )$ 
\emph{spans} a unique finite $\TI$-PDB. This is no longer the case in a 
countably infinite setting:

\begin{example}\label{ex:ti_inadmissible}
	Let $\Facts \subseteq \Facts[ \tau, \UU ]$ be countably infinite and let
	$P \from \Facts \to [0, 1] \with f \mapsto \frac12$. Then there exists no
	$\TI$-PDB with marginals according to $P$: Assume that $\D = ( \DB, P )$ is
	a $\TI$-PDB with marginals $P$ and let $D \in \DB$ be arbitrary. Suppose
	that the facts in $\Facts$ are $f_1, f_2, \dots$ and that $D \subseteq
	\set{f_1,\dots,f_k}$ for some $k \in \NN$. Recall that as the facts are
	independent, so are their complements. Then for all $n \geq 1$ we have 
	\begin{equation*}
		P\big( \set{ D } \big)
		\leq \Pr_{ D \sim \D } \big( f_{k+1}, \dots, f_{k+n} \notin D \big)
		= \prod_{ i = k+1 }^{ k+n } \big( 1 - P( f_i ) \big)
		= \tfrac{1}{2}^n\text,
	\end{equation*}
	implying $P\big( \set{ D } \big) = 0$. But then 
	\begin{equation*}
		1 = P( \DB ) = \sum_{ D \in \DB } P\big( \set{ D } \big) = 0\text,
	\end{equation*}
	a contradiction.

	We take the opportunity to comment on two subtleties with infinite PDBs in
	the light of this example:
	\begin{enumerate}
		\item In the above situation we have $\Expectation_{\D}( \size{\?} ) =
			\sum_{ f \in \Facts } P(f) = \infty$. We recall that this is, by
			itself, no contradiction to the requirement that all instances of the
			PDB be finite, see \cref{ex:infiniteexpectation}.
		\item We argued that every single database instance has probability $0$,
			and therefore, what we have is no probability space. This argument is
			only possible because $\DB$ is countable, because then the probability
			of every event is given by the sum of the probabilities of the
			elementary events $\set{ D }$ (due to $\sigma$-additivity). This does
			not hold for uncountable PDBs, such as the one sketched in
			\cref{ex:temprec}.\qedhere
	\end{enumerate}
\end{example}

In case of existence however, the $\TI$-PDB spanned by $\Facts$ and $P$ is
unique:

\begin{proposition}\label{pro:ti_unique}
	Let $\Facts \subseteq \Facts[ \tau, \UU ]$ be a countable set of facts and
	let $P \from \Facts \to [0, 1]$ with the property that there exists a 
	$\TI$-PDB $\D$ over $\Facts$ with marginals $P$. Then $\D = ( \DB, P_{ \D } 
	)$ is uniquely determined by $P$ and it holds that
	\begin{equation}\label{eq:ti_single_instance}
		P_{ \D }\big( \set{ D } \big) =
		\prod_{ f \in D } P( f )
		\cdot \prod_{ f \in \Facts \setminus D } \big( 1 - P( f ) \big)
	\end{equation}
	for all $D \in \DB = \powerset_{\fin}( \Facts )$. 
\end{proposition}

Note that the infinite product $\prod_{ f \in \Facts \setminus D }
\big( 1 - P( f ) \big)$ in \eqref{eq:ti_single_instance} is well-defined, 
because all its factors are from the interval $[0,1]$ (see \cref{sec:infprod}).

\begin{proof}
	If $\Facts$ is finite, then the statement is clear. Let $\Facts$ be 
	countably infinite and let $D \in \DB$. Suppose that $\Facts = \set{ f_1,
	f_2, \dots }$ where $f_i \neq f_j$ for all $i \neq j$. For all $i = 1, 2,
	\dots$ let $\AAAA_i$ be the set of database instances $D' \in \DB$ with the
	property that $f_j \in D'$ if and only if $f_j \in D$ for all $j \leq i$.
	That is, $\AAAA_i$ is the set of instances that coincide with $D$ when
	restricted to $\set{ f_1, f_2, \dots, f_i }$. Then it holds that $\AAAA_1
	\supseteq \AAAA_2 \supseteq \dots$ and, moreover, that $\bigcap_{ i = 1 }^{
	\infty } \AAAA_i = \set{ D }$. Since $\D$ is tuple-independent with
	marginals $P$, it holds that 
	\[
		P_{\D}\big( \AAAA_i \big) = 
		\prod_{ \substack{ j \in \set{1, \dots, i } \colon \\ f_j \in D } } 
			P(f_j) \cdot
		\prod_{ \substack{ j \in \set{1, \dots, i } \colon \\ f_j \notin D } }
			\big( 1 - P(f_j) \big)
		\text.
	\]
	Note that here we used, in particular, that in an independent family of
	events, the events are still independent when arbitrarily many of them are
	replaced by their complement. As $D$ is finite, it then follows from
	$\AAAA_1 \supseteq \AAAA_2 \supseteq \dots$ that
	\[
		P_{\D}\big( \set{ D } \big) =
		P_{\D}\bigg( \bigcap_{ i = 1 }^{ \infty } \AAAA_i \bigg) =
		\prod_{ \substack{ j \in \NN_{>0} \colon \\ f_j \in D } } 
			P( f_j )
			\cdot
			\prod_{ \substack{ j \in \NN_{>0} \colon \\ f_j \notin D } }
			\big( 1 - P( f_j ) \big)
	\]
	(see \cref{fac:semicontinuity}). That is, \cref{eq:ti_single_instance} holds
	for all $D \in \DB$. Due to the $\sigma$-additivity of $P_{\D}$, and since
	$\DB = \powerset_{\fin}( \Facts)$, \cref{eq:ti_single_instance} implies that
	$P_{\D}$ is uniquely determined.
\end{proof}

We denote the unique $\TI$-PDB over $\Facts$ with marginal probabilities $P$ 
(in case of existence) by $\D_{ \spn{ \Facts, P } }$. This notation is slightly
redundant as $\Facts$ is already provided by the domain of $P$ but we deem it
beneficial to make the set of facts explicit. We say that $\D_{ \spn{ \Facts, 
P } }$ is \emph{spanned} by $\Facts$ and $P$. Slightly abusing notation, we
also denote the probability measure of $\D_{ \spn{ \Facts, P } }$ by $P$. 

It remains to characterize which marginal probability assignments span
countably infinite $\TI$-PDBs.

\begin{theorem}\label{thm:ti_series}
	Let $\Facts \subseteq \Facts[ \tau, \UU ]$ be a countable set of facts and
	let $P \from \Facts \to [0, 1]$. Then there exists a unique $\TI$-PDB 
	spanned by $\Facts$ and $P$ if and only if
	\begin{equation*}
		\sum_{ f \in \Facts } P( f ) < \infty
		\text.\qedhere
	\end{equation*}
\end{theorem}

\begin{proof}
	The result is trivial for finite $\Facts$. Thus, we let $\Facts$ be 
	countably infinite and fix an enumeration $f_1, f_2, \dots$ of $\Facts$.
	We start with the direction from left to right. Suppose $\D = (\DB, P)$ is
	the $\TI$-PDB over $\Facts$ with marginals $P$ and assume that 
	$\sum_{ i = 1 }^{ \infty } P\big( f_i \big) = \infty$. Because every
	instance in $\DB = \powerset_{ \fin }( \Facts )$ is a \emph{finite} set of
	facts, it holds that
	\begin{equation*}
		\bigcap_{ i = 1 }^{ \infty } \bigcup_{ j = i }^{ \infty }
		\set[\big]{ D \in \DB \with f_j \in D } = \emptyset\text.
	\end{equation*}
	But then, according to the \bcref{},
	\begin{equation*}
		0 
		= P( \emptyset ) 
		= P\bigg( \bigcap_{ i = 1 }^{ \infty } \bigcup_{ j = i }^{ \infty }
			\set[\big]{ D \in \DB \with f_j \in D } \bigg)
		= 1\text,
	\end{equation*}
	a contradiction. Thus, $\sum_{ i = 1 }^{ \infty } P\big( f_i \big) <
	\infty$.

	\bigskip
	For the other direction, let $P \from \Facts \to [0, 1]$ such that 
	$\sum_{ f \in \Facts } P( f ) < \infty$. Again, let $f_1,f_2, \dots$ be
    an enumeration of $\Facts$. For all $D \in
	\DB \coloneqq \powerset_{ \fin }( \Facts )$ we define
	\begin{equation*}
		P\big( \set{ D } \big)
		= \prod_{ f \in D } P( f ) \cdot
			\prod_{ f \in \Facts \setminus D } \big( 1 - P( f ) \big)
		\text.
	\end{equation*}
	Then $P( \set{ D } ) \in [0, 1]$ is well-defined. Moreover, $P$ uniquely
	extends to a measure on $(\DB, \powerset( \DB ))$ defined by $P( \DDDD ) =
	\sum_{ D \in \DDDD } P( \set{ D } )$. Abusing notation, we denote
	this extended measure by $P$ and consider the measure space $(\DB, P)$.
	
	It remains to show that $P$ is a probability measure and has the correct
	marginals.  Now for $n = 0, 1, 2, \dots$, consider $\DB_n \coloneqq
	\powerset\big( \set{ f_1, \dots, f_n } \big)$, i.\,e. the set of instances
	made up exclusively from the facts in $\Facts_n \coloneqq \set{ f_1, \dots,
	f_n }$.  Note that for all $n \geq 0$, $\DB_n \subseteq \DB_{ n + 1 }$ and
	that $\bigcup_{ n = 0 }^{ \infty } \DB_n = \DB$, so in particular
	$P(\DB) = \lim_{n\to\infty} P( \DB_n )$ (cf. \cref{fac:semicontinuity}). We
	have
\begin{align*}
		P\big( \DB_n \big) 
		&= P \big( \set{ D \in \DB \with D \subseteq \Facts_n } \big)
		\\
		&= \sum_{ F \subseteq \Facts_n }
			\prod_{ f \in F } P( f ) \cdot
			\prod_{ f \in \Facts_n \setminus F } \big( 1 - P( f ) \big) \cdot
			\prod_{ f \notin \Facts_n } \big( 1 - P( f ) \big)
		\\
		&= \prod_{ f \notin \Facts_n } \big( 1 - P( f ) \big) \cdot
			\underbrace{
				\bigg( \sum_{ F \subseteq \Facts_n } \prod_{ f \in F } P( f ) \cdot
					\prod_{ f \in \Facts_n \setminus F } \big( 1 - P( f ) \big)
				\bigg)
			}_{ = 1 }
			\\
		&= \prod_{ f \notin \Facts_n } \big( 1 - P( f ) \big)\text.
	\end{align*}
	By \eqref{eq:prodsum}, it follows that
	\begin{equation*}
		P( \DB )
		= \lim_{ n \to \infty } P\big( \DB_n \big)
		= \lim_{ n \to \infty } \prod_{ f \notin \Facts_n }\big( 1 - P( f ) \big)
		\geq 1 - \lim_{ n \to \infty } \sum_{ f \notin \Facts_n } P( f )
		= 1\text.
	\end{equation*}
	As $P( \DB_n ) = \prod_{ f \notin \Facts_n }\big( 1 - P( f ) \big) \leq 1$
	for all $n \geq 0$, this implies $P( \DB ) = 1$.

	Thus, $P$ is a probability measure on $\DB$. It remains to show that
	$(\DB, P)$ has the correct marginals. Recalling that $\DB=\powerset_{\fin}(
   \Facts)$, we see that
	\begin{align*}
		\Pr_{ D \sim \D} \event{ f \in D } 
		&= \sum_{ \substack{ D \in \DB \with\\
				f \in D } } 
			\prod_{ g \in D } P( g ) \cdot
			\prod_{ g \in \Facts \setminus D } \big( 1 - P( g ) \big)\\
		&= P( f ) \cdot
			\bigg(
				\sum_{ D \in \powerset_{\fin}( \Facts \setminus \set{ f } ) }
				\prod_{ g \in D } P( g ) \cdot
				\prod_{ g \in \Facts \setminus D } \big( 1 - P( g ) \big)
			\bigg)
			\cdot \smash{
				\underbrace{ \Big( P( f ) + \big( 1 - P( f ) \big) \Big) }_{ = 1 }
			}\\
		&= P( f ) \cdot \sum_{ D \in \DB } P\big( \set{ D } \big)\\
		&= P( f )\text,
	\end{align*}
	as desired.
\end{proof}

Recall from \cref{sss:expected_instance_size} that $\sum_{ f \in \Facts( \D ) }
P( f )$ is exactly the expected instance size $\Expectation_{ \D }\big(
\size{ \? } \big)$ in $\D$. This immediately yields the following:

\begin{corollary}\label{cor:ti_finite_expectation}
	If\/ $\D = ( \DB, P )$ is a $\TI$-PDB then
	\begin{equation*}
		\Expectation_{ \D }\big( \size{ \? } \big)
		= \sum_{ D \in \DB } \size{ D } \cdot P\big( \set{ D } \big) 
		< \infty
		\text.\qedhere
	\end{equation*}
\end{corollary}

In particular, if $\D$ has \emph{infinite} expected instance size, then $\D$ is
\emph{not} tuple-independent. We remark that \cref{cor:ti_finite_expectation}
can be generalized to show that countable $\TI$-PDBs have \emph{all} moments 
$\Expectation_{ \D }\big( \size{ \? }^{ k }\big)$ of the instance size random
variable finite
\cite{Carmeli+2021}.

\bigskip

A nice structural property of $\TI$-PDBs is their modular setup: When 
conditioning a $\TI$-PDB on a smaller set of facts, one again obtains a 
$\TI$-PDB.

Let $\D = (\DB, P)$ be a PDB and let $\DDDD \subseteq \DB$ such that $P( \DDDD 
) > 0$. Then $\D \under \DDDD$ is the PDB with sample space $\DDDD$ and
probability measure $P_{ \DDDD } \from \DDDD \to [0, 1]$ with
\begin{equation*}
	P_{ \DDDD }\big( \set{ D } \big) 
	= P\big( \set{ D } \under \DDDD \big)
	= \frac{ P\big( \set{ D } \big) }{ P( \DDDD ) }
\end{equation*}
for all $D \in \DDDD$.

\begin{lemma}\label{lem:restr_ti}
	Let $\D = \D_{ \spn{ \Facts, P } }$ be a\/ $\TI$-PDB. Let $\Facts' \subseteq
	\Facts$ such that $P\big( \powerset_{ \fin }( \Facts' ) \big) > 0$ and let 
	$P'$ denote the restriction of\/ $P$ to\/ $\Facts'$.
       
	Then it holds that
	\begin{equation*}
		\D \under \powerset_{ \fin }( \Facts' )
		= \D_{ \spn{ \Facts', P' } }
		\text.
	\end{equation*}
	In particular, $\D \under \powerset_{ \fin }( \Facts' )$ is a $\TI$-PDB.
\end{lemma}

\begin{proof}
	Let $\DDDD = \powerset_{ \fin }( \Facts' )$, let $\D = ( \DB, P )$ and let
	$\D \under \DDDD = ( \DB_{ \DDDD }, P_{ \DDDD } )$. It follows from
	\cref{pro:ti_unique} that $P( \DDDD ) = \prod_{ f \in \Facts \setminus
	\Facts' } \big( 1 - P( f ) \big)$ and thus, for all $D \in \DDDD$ it holds
	that
	\begin{equation*}
		P_{ \DDDD }\big( \set{ D } \big)
		= \frac{ P\big( \set{ D } \big) }{ P( \DDDD ) }
		= \frac{
				\prod_{ f \in D } P( f ) 
				\cdot \prod_{ f \in \Facts \setminus D } \big( 1 - P( f ) \big)
			}{
				\smashoperator{\prod_{ f \in \Facts \setminus \Facts' }} 
					\big( 1 - P( f ) \big)
			}
		= \prod_{ f \in D } P( f ) \cdot 
			\mkern6mu\smashoperator{\prod_{ f \in \Facts' \setminus D }}\mkern6mu 
			\big( 1 - P( f ) \big)
		= P'\big( \set{ D }\big)
		\text.
	\end{equation*}
	The claim then follows, since $\D_{ \spn{\Facts', P'}}$ is unique by
	\cref{pro:ti_unique}.
\end{proof}

The aforementioned \enquote{modularity} of $\TI$-PDBs will be the central
point of view of \cref{ss:ti_via_sp} where we investigate this property in
detail.

\begin{remark}
In \cite[Definition 4.1]{GroheLindner2019}, we introduced the following
definition of tuple-independence in infinite PDBs: A countable PDB $\D$ be
called tuple-independent if for all families $\F$ of pairwise disjoint
(measurable) sets of facts it holds that $\big( \event{ D \cap \Facts \neq 
\emptyset } \big)_{ \Facts \in \F }$ is independent, where $D \sim \D$.  For
countable PDBs, this is equivalent to \cref{def:ti_pdb} (see \cite[Lemma
4.2]{GroheLindner2019}; on a high level, this easily follows from \cite[Theorem
2.13(iv)]{Klenke2014}). Definition 4.1 in \cite{GroheLindner2019}
only yields benefit when discussing uncountable PDBs, as there the marginal
events associated to individual facts do not suffice to describe the
probability spaces. We revisit the more general definition when we discuss
uncountable PDBs in \cref{s:beyond}.
\end{remark}

\subsection{Countable Superpositions}
\label{ss:sp}

In this and the following section we present a more abstract view on the
mechanisms of independence assumptions by explaining $\TI$-PDBs via
superpositions (see \cref{sss:fin_sp}). The significance of this is that it 
allows us to discuss classes of PDBs that are constructed from independent
parts already with the properties of superpositions at hand.

In this section in particular, we start by extending the discussion from
\cref{sss:fin_sp} to countable superpositions with a strong focus on their
application for PDBs.

Note that the results presented in this subsection are special cases of
well-known propositions in point process theory \cite{DaleyVere-Jones2003,%
DaleyVere-Jones2008,LastPenrose2017}---our contribution here is to
adapt these mechanisms for casting probabilistic databases into a unified
framework.

\bigskip

Throughout this subsection, we fix a family $( \D_i )_{ i \in \NN_{ > 0 } }$ of
PDBs $\D_i = ( \DB_i, P_i ) \in \ctblPDB$. Recall that by \cref{ass:ctbl1,%
ass:ctbl2} from the beginning of the section, we have that $\DB_i \subseteq
\DB[ \tau, \UU ]$ for all $i = 1, 2, ...$ and some common database schema
$\tau$ and countable universe $\UU$. Yet, we temporarily deviate from
\cref{ass:ctbl3} and allow the occurring PDBs to be \emph{bag} PDBs. 

We now present an example that can serve as a running example
throughout this section and that nicely exhibits the way we present
$\TI$-PDBs using superpositions in \cref{ss:ti_via_sp}.

\begin{example}[A Coin Flip PDB]\label{exa:coinflip}
	We want to accentuate the intuition on how a tuple-independent PDB is made
	up from a large number of small independent events by connecting them to a 
	coin flip experiment. For this, suppose that each fact in a $\TI$-PDB acts
	as a biased coin. If the coin flips head, it means the fact is present, and
	if it flips tail, the fact is omitted.\footnote{If the number of coins is 
	finite (as would be the analogy for finite $\TI$-PDBs), such experiments
	are called \emph{Poisson trials} in probability theory literature \cite[p.
	218]{Feller1968} and the distribution of the number of successes is called
	\emph{Poisson-binomial distribution} or \emph{Poisson's binomial
	distribution}, cf.~\cite{Wang1993}. Here, we consider infinitely many
	coins.}

	We cast this intuition into a concrete $\TI$-PDB of coin flips. Suppose we
	flip a countably infinite number of independent, biased coins such that the
	$i$th coin flips head with probability $p_i$. Let $\tau$ be the relational
	schema consisting of a single relation $H$ of arity $1$ and let $\UU = \NN$.
	We interpret the fact $f_i \coloneqq H(i)$ as the event that the $i$th coin
	comes up heads. Note that $\Facts \coloneqq \Facts[\tau,\UU] = \set{ f_i\
	\with i\in\NN }$. We define $P_i \from \Facts \to [0,1]$ by $P_i(f_i)
	\coloneqq p_i$ and $P_i(f)\coloneqq 0$ for $f\neq f_i$, and we let $\D_i =
	\D_{ \spn{ \Facts, P_i } }$. 

	Our goal is to express the $\TI$-PDB $\D_{ \spn{ \Facts,P }}$ with $P$
	defined by $P(f_i) = P_i(f_i) = p_i$ as an independent superposition of the
	$\D_i$. This intuitively corresponds to the observation that we obtain the
	same probability space, regardless of whether we flip the coins 
	individually, or all coins together as a whole.
\end{example}

We write $\DB^{ \otimes }$ for the set of all sequences $\vec D=( D_1, D_2,
\dots )$ where $D_i\in\DB_i$, and consider the product measure space
$\bigotimes_{ i = 1 }^{ \infty } \D_i = ( \DB^{ \otimes }, \DDD^{ \otimes },
P^{ \otimes } )$ where $\DDD^{ \otimes } = \bigotimes_{ i = 1 }^{ \infty }
\powerset( \DB_i )$. Note that this is a probability space over sequences of
database instances. We now use this space to construct a measure space on
database instances. Let $\DB_{ \fin }^{ \otimes }$ denote the set of sequences
$( D_1, D_2, \dots ) \in \DB^{ \otimes }$ where \emph{all but finitely many}
$D_i$ are the empty instance. Note that $\DB_{ \fin }^{ \otimes }$ only
contains countably many sequences.\footnote{As the $\DB_i$ are countable, for
every $n \in \NN$ it holds that the set $\DB_{\fin}^{\otimes}$ contains only
countably many sequences where exactly $n$ instances are non-empty. The
(countable) union over all these sets of instances is $\DB_{\fin}^{\otimes}$.}
For all $( D_1, D_2, \dots ) \in \DB_{ \fin }^{ \otimes }$ it holds that
\begin{equation}\label{eq:singletons_measurable}
	\set[\big]{ ( D_1, D_2, \dots ) }
	= \bigcap_{ i = 1 }^{ \infty } 
		\pi_{ i }^{ -1 }\big( \set{ D_i } \big) \in \DDD^{ \otimes }
\end{equation}
where $\pi_i$ denotes the canonical projection to the $i$th
component.
As
$\DB_{ \fin }^{ \otimes }$ is countable, \eqref{eq:singletons_measurable}
implies $\DDD_{ \fin }^{ \otimes } \coloneqq \powerset\big( \DB_{ \fin }^{ 
  \otimes } \big) \subseteq \DDD^{ \otimes }$.
Thus, $\D_{ \fin }^{ \otimes }
\coloneqq \big( \DB_{ \fin }^{ \otimes }, \powerset\big( \DB_{ \fin }^{ \otimes 
} \big), P_{ \fin }^{ \otimes }\big)$, where $P_{ \fin }^{ \otimes }$ is the
restriction of $P^{ \otimes }$ to $\DDD_{ \fin }^{ \otimes }$, is a measure
space. Note that it is not necessarily a probability space.

\begin{example}\label{exa:coinflip2}
  Let us revisit the setting of \cref{exa:coinflip}. To define the product
  measure, we consider events $\AAAA$ of the form $\prod_{ i = 1 }^{ \infty } 
  \AAAA_i$, where $\AAAA_i \neq \DB_i = \set{ \emptyset, \set{ f_i } }$ for
  only finitely many $i$ and define
  \begin{equation}
    \label{eq:prodmass}
    P^\otimes( \AAAA )
	 = \prod_{\substack{i \in \NN_{ >0 } \with \!\!\\ \AAAA_i = \set{\set{f_i}}}}p_i
	 	\cdot 
		\prod_{\substack{i \in \NN_{ >0 } \with \!\!\\ \AAAA_i = \set{\emptyset}}}(1-p_i)
	 \text.
  \end{equation}
  The sets $\AAAA$ generate the product $\sigma$-algebra $\DDD^\otimes$ on
  $\DB^\otimes$, and \eqref{eq:prodmass} uniquely determines the product
  measure $P^\otimes$. The definition \eqref{eq:prodmass} implies that for all
  $\vec D = (D_1, D_2, \ldots)\in\DB^\otimes$ we have 
  \begin{equation}\label{eq:prodmass2}
	  P^\otimes \big( \set{ \vec D } \big)
	  = \prod_{ \substack{ i \in \NN_{ >0 } \with \!\!\\D_i = \set{f_i} } } p_i
		\cdot
		\prod_{ \substack{ i \in \NN_{ >0 } \with \!\!\\D_i = \emptyset} } 
			(1 - p_i)
		\text.
  \end{equation}
	 It may well be that $P^\otimes \big( \set{ \vec D } \big) = 0$ for all
	 $\vec D \in \DB^\otimes$. By an application of the \bcref{} or by a direct
	 calculation, it can be shown that $P^\otimes \big( \set{ \vec D } \big) =
	 0$ for all $\vec D \in \DB_\fin^\otimes$ if and only if 
	 $\sum_{ i = 1 }^\infty p_i = \infty$.
\end{example}

Recall that $\powerbag( S )$ denotes the set of all bags over some set $S$.
Consider the additive union 
\begin{equation}\label[function]{eq:g_uplus}
	(\?)^\uplus \from \DB^{ \otimes } \to 
	\powerbag\bigg( \bigcup_{ i = 1 }^{ \infty } \Facts_i \bigg)
	\with
	( D_1, D_2, \dots ) \mapsto 
	( D_1, D_2, \dots )^{ \uplus } \coloneqq
	\biguplus_{ i = 1 }^{ \infty } D_i\text,
\end{equation}
where $\Facts_i \coloneqq \Facts( \D_i ) \subseteq \Facts[ \tau, \UU ]$ for all
$i = 1, 2, \dots$. 

We note that the construction presented above would work exactly the same way, 
if the PDBs $\D_1, \D_2, \dots$ were uncountable. The only real change is that
$\DDD^{\otimes}$ then becomes the product $\sigma$-algebra of the
$\sigma$-algebras of $\D_1,\D_2,\dots$ (and needing that the singleton
containing the empty instance is measurable in all of them). Due to our
application on countable PDBs here (and some intricacies in the uncountable
with respect to set semantics), we refrain from this greater generality.

\begin{lemma}\label{lem:g_uplus_db}
	Let\/ $\vec{ D } \in \DB^{ \otimes }$. 
	\begin{enumerate}
		\item The bag $\vec{ D }^{ \uplus }$ is a database instance if and
			only if\/ $\vec{ D } \in \DB_{ \fin }^{ \otimes }$.
			\label{itm:g_uplus_db_i}
		\item If\/ $\vec{ D }^{ \uplus }$ is a database instance, then it is
			a set instance if and only if\/ $\vec{ D }$ consists of pairwise
			disjoint set instances.
			\qedhere
	\end{enumerate}
\end{lemma}

\begin{proof}
  \begin{enumerate}
  \item If
    $\vec{ D } = ( D_1, D_2, \dots ) \in \DB_{ \fin }^{ \otimes }$,
    then $D_i \neq \emptyset$ for at most finitely many
    $i = 1, 2, \dots$.  Thus, $\vec{ D }^{ \uplus }$ is a database
    instance.
    
    If, on the contrary,
    $\vec{ D } = ( D_1, D_2, \dots ) \in \DB^{ \otimes } \setminus
	 \DB_{ \fin }^{ \otimes }$, then $D_i \neq \emptyset$ for infinitely many $i
	 \in \NN_{ > 0 }$. Thus, $\size{ \vec{ D }^{ \uplus } } = \infty$. In
	 particular, $\vec{ D }^{ \uplus }$ is not a database instance.  

  \item Now suppose $\vec{ D }^{ \uplus }$ is a database
    instance. Then by \labelcref{itm:g_uplus_db_i},
    $\vec{ D } \in \DB_{ \fin }^{ \otimes }$. On the one hand, if
    $D_1, D_2, \dots$ are pairwise disjoint set instances among which
    only finitely many are non-empty, then their additive union is
    clearly a set instance as well. If, on the other hand,
    $D_1, D_2, \dots$ contains either a proper bag instance or at
    least two of its components have a common fact, then
    $\vec{ D }^{ \uplus }$ is a proper bag instance.
	 \qedhere
  \end{enumerate}
\end{proof}

We now define $\biguplus_{ i = 1 }^{ \infty } \D_i$ as the image measure space
of $\D_{ \fin }^{ \otimes }$ under $( \? )^{ \uplus }$. That is, 
$\biguplus_{ i = 1 }^{ \infty } \D_i = ( \DB, P )$ with
\begin{align*}
	\DB 
	&= \powerbag_{ \fin } \Bigg( \bigcup_{ i = 1 }^{ \infty } \Facts_i \Bigg)
	\quad\text{and}\\
	P\big( \set{ D } \big)
	&= P^{ \otimes }\big( \set{ \vec{ D } \in \DB_{ \fin }^{ \otimes } \with
	\vec{ D }^{ \uplus } = D } \big)
	\text.
\end{align*}
Because the restriction of $( \? )^{ \uplus }$ to $\DB_{ \fin }^{ \otimes }$ is
measurable, this also makes $\biguplus_{ i = 1 }^{ \infty } \D_i$ a measure
space. Again, however, it is not necessarily a probability space.

\begin{example}\label{exa:coinflip3}
	We continue from \cref{exa:coinflip2}.  As in this example, for $i\neq j$
	and $D_i\in\DB_i$, $D_j\in\DB_j$ we have $D_i\subseteq\{f_i\}$,
	$D_j\subseteq\{f_j\}$ and therefore $D_i\cap D_j=\emptyset$, by
	Lemma~\ref{lem:g_uplus_db} all instances in $\biguplus_{ i = 1 }^{ \infty }
	\D_i$ are set instances, and we have $\DB=\powerset_{ \fin }( \Facts)$. Note
	that for every instance $D\in\DB$ the set $\{\vec D\mid \vec D^\uplus=D\}$
	consists of a single sequence $\vec D=(D_1,D_2,\ldots)$ with $D_i=\{f_i\}$
	if $f_i\in D$ and $D_i=\emptyset$ otherwise. Thus by \eqref{eq:prodmass2},
	\begin{equation*}
		P(\set{ D })
		= \prod_{\substack{ i \in \NN_{ >0 } \with \!\!\\ f_i \in D}} p_i \cdot
		\prod_{\substack{ i \in \NN_{ >0 } \with \!\!\\f_i\not\in D}} (1-p_i)
		\text.
  \end{equation*}
  Note that this is exactly the probability $\{D\}$ would have in a
  $\TI$-PDB with fact probabilities $P(f_i)=p_i$. Of course such a $\TI$-PDB
  only exists if $\sum_i p_i<\infty$.
\end{example}

\begin{remark}
	If the measure space $\biguplus_{ i = 1 }^{ \infty } \D_i$ is a probability
	space, then it is a bag PDB. If this is the case, but additionally all 
	$\D_i$ are set PDBs with disjoint fact sets, then we can proceed as follows
	to really obtain \emph{set} PDBs instead of simple bag PDBs: We treat $( \?
	)^{ \uplus }$ as a function into $\powerset\big( \bigcup_{ i = 1 }^{ \infty
	} \Facts_i \big)$ and let $\DB = \powerset_{ \fin }\big( \bigcup_{ i = 1 }^{
	\infty } \Facts_i \big)$. Then $\biguplus_{ i = 1 }^{ \infty } \D_i$ is a
	set PDB.
\end{remark}

Generalizing the previous example, our tentative goal is to express arbitrary
$\TI$-PDBs as the superposition of single fact PDBs. In particular, the
superposition should have the desired independence properties, and the correct
marginal probabilities. This is prepared with the following lemma.

\begin{lemma}\label{lem:sp_ind}
	Suppose $\D \coloneqq \biguplus_{ i = 1 }^{ \infty } \D_i = ( \DB, P )$ is
	a PDB. For all $i = 1,2, \dots$ let $\DDDD_i \subseteq \DB_i$ such that $D
	\cap \Facts_j = \emptyset$ for all $D \in \DDDD_i$ and all $j \neq i$.
	Let 
	\[
		\widetilde{ \DDDD }_i 
			\coloneqq
		\DB_{\fin}^{\otimes} \cap \pi_i^{-1}( \DDDD_i )
			= 
		\biguplus \bigg( \DB_{ \fin }^{ \otimes } \cap \Big( \prod_{ j = 1 }^{ i-1 }
			\DB_j \times \DDDD_i \times \prod_{ j = i+1 }^{ \infty } \DB_j \Big) \bigg)
		\text.
	\]
	Then
	\begin{enumerate}
		\item $P( \widetilde{\DDDD}_i ) = P_i( \DDDD_i )$ for all
			$i = 1, 2, \dots$, and
		\item the family $( \widetilde{\DDDD}_i )_{ i \in \NN_{>0} }$ is
			independent in $\biguplus_{ i = 1 }^{ \infty } \D_i$.
		\qedhere
	\end{enumerate}
\end{lemma}

Note that the precondition of \cref{lem:sp_ind} states that the instances in
$\DDDD_i$ do not appear among the instances of the other PDBs $\D_j$, nor can
they be assembled from them (apart from the empty instance).  That is, the
preimage of $\widetilde{\DDDD}_i$ under $( \? )^{\uplus}$ is exactly
$\DB_{\fin}^{\otimes} \cap \big(\prod_{ j = 1 }^{ i - 1 } \DB_j \times \DDDD_i
\times \prod_{ j = i + 1 }^{\infty} \DB_j\big)$.

\begin{proof}\leavevmode
	\begin{enumerate}
	\item Let $i \in \NN_{ > 0 }$. Then
		\begin{align*}
			P\big( \widetilde{ \DDDD }_i \big) 
			&= P^{ \otimes }\big( \set{ \vec D \in \DB_{\fin}^{\otimes} \with
			\vec D^{\uplus} \in \widetilde{ \DDDD }_i } \big)\\
			&= P^{ \otimes }\big( \set{ (D_1,D_2,\dots) \in \DB_{\fin}^{\otimes} \with D_i \in \DDDD_i } 
			\big)\\
			&= P^{ \otimes }\big( \DB_{\fin}^{\otimes} \cap \pi_{i}^{-1}( \DDDD_i
			) \big) = P^{\otimes}\big( \pi_i^{-1}( \DDDD_i ) \big) = P_i( \DDDD_i
			)
			\text.
		\end{align*}
		For the second equality, we used the property discussed right before the
		proof. The last line uses that $P^{\otimes}( \DB^{\otimes} \setminus
		\DB_{\fin}^{\otimes} ) = 0 $.
  \item Let $k \in \NN_{ > 0 }$ and let $i_1, \dots, i_k \in \NN_{> 0}$. Let 
    $\biguplus_{ i = 1 }^{ \infty } \D_i = (\DB, P)$. It holds that
    \begin{equation*}
      P
      \Bigg( 
      \bigcap_{ j = 1 }^{ k } 
		\widetilde{ \DDDD }_{ i_j }
      \Bigg)
      = P^{ \otimes }
      \Bigg(
      \bigcap_{ j = 1 }^{ k } 
      \pi_{ i_j }^{ -1 }\big( \DDDD_{ i_j } \big)
      \Bigg)
      = \prod_{ j = 1 }^k P_{i_j}( \DDDD_{i_j} )
      = \prod_{ j = 1 }^k 
		P \big( \widetilde{ \DDDD }_{i_j} \big)
      \text,
    \end{equation*}
	 where the first equality again uses the property from before the proof, and
	 the middle equality is due to the properties of the product measure
	 (\cref{fac:product_measure}). Because the above holds for arbitrary $k$ and
	 $i_1,\dots,i_k$ it follows that $\big( \widetilde{ \DDDD }_{ i } \big)_{i
	 \in \NN_{ > 0 } }$ is independent in $\biguplus_{ i = 1 }^{ \infty } \D_i$.\qedhere
  \end{enumerate}
\end{proof}

We close the discussion of countable superpositions with the probabilistic
counterpart of \cref{lem:g_uplus_db}.

\begin{lemma}\label{lem:sp_pdb}\leavevmode
	\begin{enumerate}
		\item The measure space\/ $\biguplus_{ i = 1 }^{ \infty } \D_i$ is a PDB
			if and only if\/ $\sum_{ i = 1 }^{ \infty } \Pr_{ D_i \sim \D_i }
			\event{ D_i \neq \emptyset } < \infty$.
			\label{itm:sp_pdb_i}
		\item If\/ $\biguplus_{ i = 1 }^{ \infty } \D_i$ is a PDB, then it is a
			set PDB if and only if all $\D_i$ are set PDBs with pairwise disjoint
			fact sets.\qedhere
	\end{enumerate}
\end{lemma}

\begin{proof}
  Before we turn to the proof, for all $i = 1, 2, \dots$ we define 
  $\NonEmpty_i$ as the set of sequences in $\DB^{ \otimes }$ where the $i$th
  instance is non-empty. That is, for all $i = 1, 2, \dots$ it holds that
  \begin{equation*}
    \NonEmpty_i 
    = \pi_i^{ - 1 }\big( \DB_i \setminus \set{ \emptyset } \big)
    = \prod_{ j = 1 }^{ i - 1 } \DB_j
    \times \big( \DB_i \setminus \set{ \emptyset } \big)
    \times \prod_{ j = i + 1 }^{ \infty } \DB_j
    \text.
  \end{equation*}
  Then by construction, $P^{ \otimes }( \NonEmpty_i ) = \Pr_{ D_i \sim
    \D_i } \event{ D_i \neq \emptyset }$ for all $i = 1, 2, \dots$. Note that
  \begin{equation}\label{eq:nonempty_bc}
    \DB_{ \fin }^{ \otimes } 
    = \DB^{ \otimes }
    \setminus
    \bigg( 
    \bigcap_{ i = 1 }^{ \infty } 
    \bigcup_{ j = i }^{ \infty }
    \NonEmpty_j 
    \bigg)\text.
  \end{equation}
  
  \begin{enumerate}
  \item First suppose that $\sum_{ i = 1 }^{ \infty } P^{ \otimes }
    (\NonEmpty_i) = \sum_{ i = 1 }^{ \infty } \Pr_{ D_i \sim \D_i }
    \event{ D_i \neq \emptyset } < \infty$. Then by the \bcref{} it follows
    from \cref{eq:nonempty_bc} that $P( \DB ) = P^{ \otimes } \big(
    \DB_{ \fin }^{ \otimes } \big) = 1$. That is, $P$ is a probability
    measure on $\DB_{ \fin }^{ \otimes }$, so $\biguplus_{ i = 1 }^{ 
      \infty } \D_i$ is a PDB.
    
    For the other direction let $\sum_{ i = 1 }^{ \infty } P^{ \otimes }
    ( \NonEmpty_i ) = \sum_{ i = 1 }^{ \infty } \Pr_{ D_i \sim \D_i }
    \event{ D_i \neq \emptyset } = \infty$. 
	 Note that the events $\NonEmpty_i$ are (in particular pairwise) independent
	by the properties of the product measure (\cref{fac:product_ind}).
 Thus, from 
    \eqref{eq:nonempty_bc} and the \bcref{} we get that
    $P^{ \otimes }\big( \DB_{ \fin }^{ \otimes } \big) = P( \DB ) = 0$ and
    therefore, $\biguplus_{ i = 1 }^{ \infty } \D_i$ is no PDB.
  \item This is a direct consequence of the second part of 
    \cref{lem:g_uplus_db}.\qedhere
  \end{enumerate}
\end{proof}

\begin{remark}
	\cref{thm:ti_series} can be regarded as a special case of
	\cref{lem:sp_pdb}\labelcref{itm:sp_pdb_i} for the $\D_i$ being single-fact
	PDBs (cf. \crefrange{exa:coinflip}{exa:coinflip3}). Comparing the proofs, we
	find that the proof \cref{lem:sp_pdb}\labelcref{itm:sp_pdb_i} is based on
	the more advanced machinery of product measures, which enables us to apply
	the \bcref{} where in the proof of \cref{thm:ti_series} we gave an ad-hoc
	construction of the finitary part of the product space. However, at their
	core the two proofs are very similar.
\end{remark}

As with finite superpositions, where it can easily be checked by hand, the 
independent superposition does not depend on the order of the involved spaces
whatsoever, because $\uplus$ is commutative. Thus, we allow superpositions
$\biguplus_{ i \in I } \D_i$ for arbitrary countable index sets $I$. Finally 
note that finite superpositions could have been introduced equivalently with a
product space construction instead of the explicit description in 
\cref{sss:fin_sp}.

\subsection{Tuple-Independence via Superpositions}
\label{ss:ti_via_sp}

Let us now carry out the argument describing $\TI$-PDBs in terms of
superpositions in detail. We return to our assumptions
\cref{ass:ctbl1,ass:ctbl2,ass:ctbl3} to the full extent. Recall that
for every $\Facts \subseteq \Facts[ \tau, \UU ]$ and
$P \from \Facts \to [0, 1]$, the (unique) $\TI$-PDB spanned by
$\Facts$ and $P$ is denoted by $\D_{ \spn{ \Facts, P } }$, provided
that it exists. We show that $\TI$-PDBs can be decomposed into
arbitrary independent components.

\begin{theorem}\label{thm:ti_decomposition}
  Let $\Facts \subseteq \Facts[ \tau, \UU ]$ be a (countable) set of facts and
  $P \from \Facts \to [0, 1]$. Suppose that $\big( \Facts_i \big)_{ i \in I }$
  is a partition of\/ $\Facts$ and let $P_i$ denote the restriction of $P$ to
  $\Facts_i$.
	\begin{enumerate}
		\item The $\TI$-PDB $\D_{ \spn{ \Facts, P } }$ exists if and only if\/
			$\biguplus_{ i \in I } \D_{ \spn{ \Facts_i, P_i } }$ is a PDB.
			\label{itm:ti_decomposition_i}
		\item If\/ $\D_{ \spn{ \Facts, P } }$ exists, then 
			\begin{equation*}
				\D_{ \spn{ \Facts, P } } 
				= \biguplus_{ i \in I } \D_{ \spn{ \Facts_i, P_i } }
				\text.\qedhere
			\end{equation*}
			\label{itm:ti_decomposition_ii}
	\end{enumerate}
\end{theorem}

\begin{proof}
	We first show the left-to-right direction of
	\labelcref{itm:ti_decomposition_i} and then handle the other direction and
	\labelcref{itm:ti_decomposition_ii} simultaneously.

	Suppose that $\D_{ \spn{ \Facts, P } }$ exists. Then by
	\cref{thm:ti_series} it holds that $\sum_{ f \in \Facts } P( f) <
	\infty$. In particular, $\sum_{ f \in \Facts_i } P_i( f ) < \infty$ for all
	$i \in I$, so by \cref{thm:ti_series}, all $\D_{ \spn{ \Facts_i, P_i
	} } \eqqcolon \D_i$ exist. Then it holds that
	\begin{align*}
		\sum_{ i \in I } \Pr_{ D \sim \D_i } \event{ D \neq \emptyset }
		&= \sum_{ i \in I } \Pr_{ D \sim \D_i } \bigcup_{ f \in \Facts_i } 
			\event{ f \in D }\\
		&\leq \sum_{ i \in I } \sum_{ f \in \Facts_i } \Pr_{ D \sim \D_i } 
			\event{ f \in D }
		= \sum_{ i \in I } \sum_{ f \in \Facts_i } P_i ( f )
		= \sum_{ f \in \Facts } P( f )
		< \infty
		\text.
	\end{align*}
	Thus, by \cref{lem:sp_pdb}, $\biguplus_{ i \in I } \D_i$ is a PDB. 

	If $\biguplus_{ i \in I } \D_{ \spn{ \Facts_i, P_i } } \eqqcolon \D$ is a
	PDB, then in particular, $\D_{ \spn{ \Facts_i, P_i } }$ is a PDB for all $i
	\in I$. By \cref{lem:sp_ind}, the marginal probability of $f \in \Facts_i$
	in $\D$ coincides with $P_i( f ) = P( f )$. Moreover, the events $\big(
	\event{ f \in D } \big)_{ f \in \Facts_i }$ are independent in $\D_{
	\spn{ \Facts_i, P_i } }$ for all $i \in I$. From \cref{lem:sp_ind} it
	follows that the events $\big( \event{ f \in D } \big)_{ f \in \Facts }$ are
	independent in $\D$. Together, $\D$ is a $\TI$-PDB with fact set $\Facts( \D
	) = \Facts$ and marginals according to $P$. By uniqueness
	(\cref{pro:ti_unique}), it follows that $\D = \D_{ \spn{ \Facts, P } }$.
\end{proof}

Now for a set of facts $\Facts$, a function $P \from \Facts \to [0 ,1]$ and 
$f \in \Facts$ we let $\D_{ f } = \big( \DB_{ f }, P_f \big)$ denote the single
fact PDB with
\begin{align*}
	\DB_f &= \set[\big]{ \emptyset, \set{ f } } \quad\text{ and }\\
	P_f\big( \set{ D } \big) &=
	\begin{dcases}
		P( f ) 
		&\text{if } D = \set{ f } \text{ and}\\
		1 - P( f )
		&\text{if } D = \emptyset
	\end{dcases}
\end{align*}
for all $D \in \set[\big]{ \emptyset, \set{ f } }$. Applying 
\cref{thm:ti_decomposition} to these single-fact PDBs yields the
desired connection between infinite $\TI$-PDBs and superpositions.

\begin{corollary}
	Let $\Facts$, $P$ and $\big( \D_f \big)_{ f \in \Facts }$ as above.
	\begin{enumerate}
		\item The $\TI$-PDB $\D_{ \spn{ \Facts, P } }$ exists if and only if\/
			$\biguplus_{ f \in \Facts } \D_f$ is a PDB.
			\label{itm:ti_span_sp_i}
		\item If\/ $\D_{ \spn{ \Facts, P } }$ exists, then
			\begin{equation*}
				\D_{ \spn{ \Facts, P } } 
				= \biguplus_{ f \in \Facts } \D_f
				\text.\qedhere
			\end{equation*}
	\end{enumerate}
\end{corollary}

Explaining $\TI$-PDBs by superpositions suggests natural generalizations
of independence assumptions for $\TI$-PDBs. In \cref{thm:ti_decomposition},
what would happen if we were to replace the PDBs $\D_{ \spn{ \Facts_i, P_i } }$
with PDBs that are not tuple-independent? In this case, we still have the
independence properties of the superposition (\cref{lem:sp_ind}). One example
of such PDBs are block-independent disjoint PDBs, which are the subject of the
next subsection.

\subsection{Block-Independence and Beyond}
\label{ss:bid}

Finite probabilistic databases are referred to as being \emph{block-independent
disjoint} if the space of facts can be partitioned into stochastically
independent \enquote{blocks} such that database instances almost surely contain
at most one fact from each block. The following example illustrates that key
constraints naturally lead to block-independent disjoint PDBs.

\begin{example}\label{ex:order2}
	We recall the setting of \cref{ex:order} of a database of orders in a single 
	relation $\REL{Order}$. Now the attribute $\ATT{OrderID}$ is the key of the
	relation. Worlds violating the key constraint, that is, containing multiple
	tuples with the same key, should have probability $0$. The grouping and the
	additional horizontal lines indicate the blocks of the PDB.

	\begin{figure}[H]
		\centering%
		\begin{tcolorbox}[hbox,fancy table,tabular={r llr r},title={$\REL{Order}$},before upper app={\rowcolor{verylightgray}}]
			\underline{OrderID}
						& Customer	& ShipTo 		& Price [\textdollar]	& $P$	  \\
			\hline\hline%
			\multirow{2}{*}{0000001}	
						& Joe			& New York 		& 99 							& 0.8	  \\
						& Bob			& Los Angeles 	& 199							& 0.2	  \\
			\hline%
			0000002	& Emma		& Austin			& 70							& 1.0	  \\
			\hline%
			\multirow{3}{*}{0000003}	
						& Dave		& Atlanta		& 19							& 0.2	  \\
						& Sophia		& Bakersfield	& 25							& 0.6	  \\
						& Isabella	& Boston			& 100							& 0.1	  
		\end{tcolorbox}
		\caption{Representation of a finite block-independent disjoint PDB.}
	\end{figure}

	The positive probability outcomes of the above PDBs have two or three facts.
	The first one is either the fact belonging to Joe, or the one belonging to
	Bob. The second fact will always be the one belonging to Emma.  The third
	fact is either one of the facts belonging to Dave, Sophia and Isabella or
	not present at all (with probability $1 - 0.2 - 0.6 - 0.1 = 0.1$).

	Recall the coin flip interpretation of a $\TI$-PDB. In a $\BID$-PDB, instead
	of flipping a coin, we independently roll a distinguished (multi-sided, 
	unfair) die per block. The outcome of a die roll determines the fact drawn
	from the block (or whether not to include a fact from the block).
\end{example}

\begin{definition}[Countable $\BID$-PDBs]\label{def:bid_pdb}
	A PDB $\D = ( \DB, P )$ is called \emph{block-independent disjoint} (or, a 
	\emph{$\BID$-PDB}) if there exists a partition $\Blocks$ of  $\Facts( \D )$ 
	such that
	\begin{enumerate}
		\item for all families $( f_{ \Block } )_{ \Block \in \Blocks }$ of facts
			from distinct blocks $\Block \in \Blocks$, the events $f_{ \Block }
			\in D$ are independent, that is,
			\[
				\Pr_{ D \sim \D } \event{ f_1, \dots, f_k \in D }
				= P( f_1 ) \cdot \dotsc \cdot P( f_k )
			\]
			for all $k = 1, 2, \dots$ where $f_1, \dots, f_k$ are from pairwise
			distinct blocks; and \label{itm:bid_pdb_i}
    \item for all $\Block \in \Blocks$, the events $f \in \Block$ are pairwise 
		disjoint, that is ,
		\begin{equation*} 
			\Pr_{ D \sim \D } \event{ f, f' \in D } = 0
		\end{equation*}
      for all $\Block \in \Blocks$ and all $f, f' \in \Block$.
      \label{itm:bid_pdb_ii}
      \qedhere
    \end{enumerate}
  \end{definition}

Note that the blocks of a $\BID$-PDB are unique up to facts of marginal 
probability $0$. The \emph{canonical block partition} of a $\BID$-PDB is the
block partition, where all facts with marginal probability $0$ are grouped in
a single distinguished block. We denote this partition $\Blocks( \D )$.
Similar to the $\TI$-PDB case, the probability space of a $\BID$-PDB is uniquely determined by
the marginal probabilities. In case of existence, we denote the $\BID$-PDB with
blocks $\Blocks$ and marginal probabilities $P$ by $\D_{ \spn{ \Blocks, P } }$.
The probability measure of this PDB is as follows:
\begin{equation}\label{eq:bid_probs}
	\Pr_{ D \sim \D_{ \spn{ \Blocks, P } } }\big( \set{ D } \big)
	=
	\begin{dcases}
		\prod_{ \substack{ \Block \in \Blocks \with \mkern-2mu \\ 
			\size{ D \cap \Block } = 1
		} } P( f_{ \Block } ) \cdot
		\prod_{ \substack{ \Block \in \Blocks \with \mkern-2mu \\
			\size{ D \cap \Block } = 0
		} } \big( 1 - {\textstyle\sum_{ f \in \Block }} P( f ) \big)
		& \text{if }\size{ D \cap \Block } \leq 1 \text{ for all } 
			\Block \in \Blocks\\
		0 & \text{otherwise,}
	\end{dcases}
\end{equation}
where in the first case $f_{ \Block }$ denotes the unique fact in $D \cap 
\Block$.

\bigskip

A (countable) PDB $\D$ is called \emph{functional}, if $\Pr_{ D \sim \D }
\event{ f, f' \in D } = 0$ for all facts $f\neq f'$.\footnote{In other words,
a PDB is functional, if all possible worlds with positive probability are
singletons.} Two PDBs $\D_1, \D_2$ are called \emph{non-intersecting} if
$\Facts( \D_1 ) \cap \Facts( \D_2 ) = \emptyset$.  We say an instance $D$
\emph{intersects} a block $\Block$, if $D$ contains a fact from $\Block$.

\begin{theorem}\label{thm:ctblbid}\leavevmode
	\begin{enumerate}
		\item Every (countable) $\BID$-PDB $\D$ is a superposition of
			non-intersecting functional PDBs, and the sum of all block
			probabilities is finite, that is
			\[
				\sum_{ \Block \in \Blocks( \D ) } \Pr_{ D \sim \D } 
				\event{ D \text{\normalfont{} intersects } \Block } 
					< \infty\text.
			\]
		\item For every family of functional, non-intersecting PDBs $\D_1, \D_2, 
			\dots$ with $\sum_{ i = 1 }^{ \infty } \Pr_{ D \sim \D_i } 
			\event{ D \neq \emptyset } < \infty$, we can construct a $\BID$-PDB
			$\D$ with the marginal probabilities from $\D_1, \D_2, \dots$, that
			is, with
			\[
				\Pr_{ D \sim \D } \event{ f \in D } = 
				\Pr_{ D \sim \D_i } \event{ f \in D }
			\]
			for all $f \in \Facts( \D_i )$ and all $i = 1, 2, \dots$.
			\qedhere
	\end{enumerate}
\end{theorem}

\begin{proof}\leavevmode
	\begin{enumerate}
		\item Suppose $\D = \D_{ \spn{ \Blocks, P } }$ where $\Blocks$ is a block
			partition of some fact set $\Facts \subseteq \Facts[\tau,\UU]$ and $P
			\from \Facts \to [0, 1]$ with $\sum_{ f \in \Block } P( f ) \leq 1$
			for all $\Block \in \Blocks$. For all $\Block \in \Blocks$, recall
			that $\D_{ \Block } \coloneqq \D \under \powerset_{ \fin }( \Block )$
			has $\Facts( \D_{ \Block } ) = \Block$ and
			\[
				\Pr_{ D \sim \D_{ \Block } } \event{ f \in D }
				= \frac{
					\Pr_{ D \sim \D } \event[\big]{ D = \set{ f } } }{
					\Pr_{ D \sim \D } \event{ D \text{ does not intersect any } 
					\Block' \neq \Block }
				}
				= P( f )
			\]
			for all $f \in \Block$ (see \eqref{eq:bid_probs}). Also, the PDBs 
			$\D_{ \Block }$ are all functional and pairwise non-intersecting.

			Because $\D$ is a $\BID$-PDB, the events $\event{ D \text{ intersects
			} \Block }$ are independent in $\D$, and it holds that
			\[
				\Pr_{ D \sim \D } \event{ D \text{ intersects } \Block 
					\text{ for infinitely many } \Block \in \Blocks } = 0
					\text.
			\]
			Thus, it follows from the \bcref{} that
			\[
				\infty
				> \sum_{ \Block \in \Blocks } 
					\Pr_{ D \sim \D } 
						\event{ D \text{ intersects } \Block }
				= \sum_{ \Block \in \Blocks } 
					\Pr_{ D \sim \D_{ \Block }} \event{ D \neq \emptyset }
				\text.
			\]
			By \cref{lem:sp_ind}, $\biguplus_{ \Block \in \Blocks } 
			\D_{ \Block }$ is block-independent disjoint over $\Facts$ with
			blocks $\Block$ and marginals according to $P$, so 
			$\biguplus_{ \Block \in \Blocks } \D_{ \Block } = \D_{ \spn{ \Blocks, 
			P } }$.
		\item Let $\D_1, \D_2, \dots$ be a family of functional, non-intersecting
			PDBs with $\sum_{ i = 1 }^{ \infty } \Pr_{ D \sim \D_i } \event{ D
			\neq \emptyset } < \infty$. Let $\Facts \coloneqq \bigcup_{ i = 1 }^{ 
			\infty } \Facts( \D_i )$ and let $\Blocks$ be the partition of 
			$\Facts$ into the sets $\Facts( \D_i )$. By \cref{lem:sp_pdb},
			$\biguplus_{ i = 1 }^{ \infty } \D_i$ is a PDB. It follows from
			\cref{lem:sp_ind}, that $\biguplus_{ i = 1 }^{ \infty } \D_i$ is 
			block-independent disjoint with blocks $\Blocks$ and such that
			\[
				\Pr_{ D \sim \biguplus_{ i = 1 }^{ \infty } \D_i }
				\event{ f \in D }
				=
				\Pr_{ D \sim \D_i } \event{ f \in D }
			\]
			for all $f \in \Facts( \D_i )$ and all $i = 1, 2, \dots$.
			\qedhere
	\end{enumerate}
\end{proof}

We remark that analogues of \cref{lem:restr_ti,thm:ti_decomposition} also hold
for $\BID$-PDBs. Regarding \cref{lem:restr_ti}, one can show that the 
restriction of $\BID$-PDBs to a subset of its blocks will again yield a 
$\BID$-PDB. As for \cref{thm:ti_decomposition}, it turns out that a $\BID$-PDB
can be decomposed into arbitrary smaller $\BID$-PDBs, each of which is made up
of a subset of the original blocks. In either case, the argumentation is
completely parallel to the respective proof for $\TI$-PDBs.

\begin{remark}[Beyond Block-Independent Disjoint PDBs.]
	The concept of superposition can be used to build PDBs from arbitrary 
	independent building blocks (provided that the condition from 
	\cref{lem:sp_pdb} is satisfied). As such, the model can be used to discuss
	classes of PDBs that are subject to independence assumptions but contain
	arbitrary correlations within the independent \enquote{blocks}. Such a
	building process may even be interleaved with other standard constructions
	like disjoint union or convex combinations. 
\end{remark}

\subsection{\texorpdfstring{$\FO$}{FO}-Definability over 
	\texorpdfstring{$\TI$}{TI}-PDBs}
\label{ss:fo}

In this section, we discuss views of tuple-independent PDBs. For
(non-probabilistic) relational databases, a \emph{view} is a function $V$ that
maps database instances of some \emph{input schema} $\tau$ to some \emph{output
schema} $\tau'$. If $\tau'$ consists of only a single relation symbol, say
$\tau' = \set{ R }$, then $V$ is called a \emph{query}. If additionally, $R$ is
$0$-ary, then $V$ is called \emph{Boolean}. A view is called \emph{$\FO$-view}
if it is expressible as a first-order formula under the standard semantics
\cite{Abiteboul+1995} (note that we only discuss set semantics here). The
following is easy to verify.

\begin{fact}\label{fac:fo_size}
	Let $\phi = \phi(x_1, \dots, x_k)$ be an $\FO$-formula with $k$ free 
	variables over $\tau$, possibly mentioning constants from $\UU$. Let $D \in 
	\DB[\tau, \UU]$ such that $\phi(D) = \set{ (a_1,\dots,a_k) \with D \models
	\phi(a_1,\dots,a_k) }$ is finite. Then 
	\[
		\phi(D) \subseteq \big( \adom(D) \cup \adom( \phi )\big)^k
	\]
	where $\adom(\phi)$ denotes the set of constants from $\UU$ appearing in
	$\phi$.
\end{fact}

Going back to PDBs, we first clarify the semantics of queries and views. Let
$\D = ( \DB, P )$ be a PDB and let $V$ be a view with input schema
$\tau$ and output schema $\tau'$ (that is, $V \from \DB[\tau, \UU] \to
\DB[\tau', \UU]$). Let $\DB' = \DB[\tau', \UU]$. Then we let $V( \D ) 
\coloneqq ( \DB', P' )$ with
\begin{equation}\label{eq:viewsem}
	P'( \DDDD' ) 
	\coloneqq P \big( V^{-1}(\DDDD') \big) 
	= P \set[\big]{ D \in \DB \with V(D) \in \DDDD' }
\end{equation}
for all $\DDDD' \subseteq \DB'$. That is, 
\[
	P'(\set{D'}) = \sum_{ D \in \DB \with V(D) = D' } P(\set{D})
\]
for all $D' \in \DB'$. (Formally, $P'$ is the push-forward, or image measure of
$P$ under $V$, cf. \cref{app:prob}.) We call $V(\D)$ the \emph{image} of $\D$
under $V$.

The semantics described above lifts single instance semantics to sets of
instances. In the PDB literature, this has been called \emph{possible answer
sets semantics} in the context of finite PDBs \cite[p.~22]{Suciu+2011}. 

\begin{remark}
	The semantics are defined in the exact same way for uncountable PDBs, but
	then we have to make sure that the query or view is a \emph{measurable}
	function of database instances (see \cref{app:prob}). Otherwise, 
	\eqref{eq:viewsem} is not well-defined.

	Another semantics that is used for finite PDBs is the so-called
	\emph{possible answer semantics} \cite{Suciu+2011}. Therein, the result of a
	view is the collection of the marginal probabilities of the individual
	tuples. Recalling \cref{ex:single_fact_pdf}, this kind of representation of
	query results is of no use in uncountable PDBs.
\end{remark}

We call a PDB $\D$ \emph{$\FO$-definable} over a PDB $\D_0$, if there exists an
$\FO$-view $V$ such that $\D = V( \D_0 )$. If $\classstyle{D}$ is a class of
PDBs, we let $\FO(\classstyle{D})$ denote the class of PDBs that are 
$\FO$-definable over some PDB in $\classstyle{D}$. Moreover, for a database 
schema $\tau$ and universe $\UU$, we let $\FO[\tau,\UU]$ denote the class of
first-order formulae over $\tau$ that are allowed to use constants from $\UU$.
One justification for tuple-independence to be a viable concept for finite PDBs
is that \emph{any} finite PDB is an $\FO$-view of a finite $\TI$-PDB
\cite{Suciu+2011}. This, however, does not extend to countably infinite PDBs.

\begin{example}
	We reconsider the PDB $\D$ of infinite expected size that was introduced in 
	\cref{ex:infiniteexpectation}: The sample space $\DB$ of $\D = (\DB, P)$ 
	consists of the instances $\set{ D_i \with i \in \NN_{ > 0 } }$ where $D_n =
	\set{ R(1), \dots, R(2^n) }$ and the probabilities are given by $P\big(\set{
	D_n }\big) = \frac{6}{ \pi^2 \cdot n^2 }$. We claim that $\D$ is not
	$\FO$-definable over any $\TI$-PDB.\par

	To obtain a contradiction, suppose that $\D = V( \D_0 )$ for some $\TI$-PDB
	$\D_0 = (\DB_0, P_0)$. By \cref{cor:ti_finite_expectation}, $\Expectation_{
	\D_0 } \big( \size{ \? } \big) < \infty$ because $\D_0$ is a $\TI$-PDB. We let
	$r$ denote the maximal arity of a relation in the schema of $\D_0$. Then
	every $\DB_0$-instance $D_0$ satisfies
	\begin{equation*}
		\size{ \adom( D_0 ) } \leq r \cdot \size{ D_0 }\text.
	\end{equation*}

	Since the schema of $\D$ consists of a single relation symbol, $V$ consists
	of a single $\FO$-formula $\phi(x) \in \FO[ \tau_0, \UU ]$ where $\tau_0$ is
	the schema of $\D_0$ and $\UU$ is some common universe underlying $\D$ and
	$\D_0$. Let $c$ be the number of constants from $\UU$ appearing in $\phi$.
	As $\phi$ has $k=1$ free variables, by \cref{fac:fo_size}, for every
	$(\tau_0,\UU)$-instance $D_0$ it holds that 
	\begin{equation*}
		\size{ V( D_0 ) } = 
		\size{ \phi( D_0 ) } \leq
		( \size{ \adom( D_0 ) } + c )^k =
		\size{ \adom( D_0 ) } + c \leq 
		r \size{ D_0 } + c
		\text.
	\end{equation*}
	Thus,
	\begin{align*}
		\Expectation_{\D} \big( \size{ \? } \big)
		= \sum_{ D \in \DB } \size{ D } \cdot P\big( \set{ D } \big)
		&= \sum_{ D \in \DB } \size{ D } \cdot P_0\big( V^{-1}( D ) \big)\\
		&= \sum_{ D \in \DB } \sum_{ D_0 \in V^{-1}(D) }
			\size{ V( D_0 ) } \cdot P_0\big( \set{ D_0 } \big)\\
		&= \sum_{ D_0 \in \DB_0 } \size{ V( D_0 ) } \cdot 
			P_0\big( \set{ D_0 } \big)\\
		&\leq \sum_{ D_0 \in \DB_0 } \big( r \cdot \size{ D_0 } + c \big) \cdot
			P_0\big( \set{ D_0 } \big)\\
		&= r \cdot \Expectation_{ \D_0 }\big( \size{ \? } \big) + c
		< \infty
		\text,
	\end{align*}
	a contradiction.
\end{example}

The example demonstrates the following proposition.

\begin{proposition}
	$\FO\big( \ctblsetTI \big) \subsetneq \ctblsetPDB$.
\end{proposition}

This means the infinite extension of $\TI$-PDBs is in some sense not as powerful
as its finite counterpart since $\FO\big( \finsetTI \big) = \finsetPDB$
\cite[Proposition 2.16]{Suciu+2011}. Incidentally, the proof of
\cite[Proposition 2.16]{Suciu+2011} cannot be translated to the infinite
setting, as it relies on an exhaustive encoding of all possible worlds into
facts of marginal probability $1$. This approach is not suitable for infinite
PDBs as it directly causes the sum of all marginal probabilities to diverge.

\subsection{Independent Completions}
\label{ss:ind_completion}

This subsection is devoted to a generalization of the idea of \emph{open-world
probabilistic databases} and \emph{$\lambda$-completions} of PDBs that was
introduced in \cite{Ceylan+2016}. 

\begin{definition}[Completions]\label{def:completion}
	Let $\Facts$ be a set of facts and let $\D = ( \DB, P )$ be a PDB with
	$\DB = \powerset_{ \fin }( \Facts )$. Let $\complete{ \DB } = \powerset_{
	\fin }\big( \Facts[\tau, \UU] \big)$.  Then a PDB $\complete{ \D } = (
	\complete{ \DB }, \complete{ P } )$ is called a \emph{completion} of $\D$ if
	$\complete{ P }( \DB ) > 0$ and
	\begin{equation*}
		\complete{ P }\big( \set{ D } \under \DB \big)
		= P\big( \set{ D } \big)\qquad
		\text{for all }D \in \DB \text.
		\qedhere
	\end{equation*}
\end{definition}

In other words, $\D'$ is a completion of $\D$ if the sample space of $\D'$ is
all of $\powerset_{ \fin }\big( \Facts[ \tau, \UU ] \big)$ and $\D$ is the
probability space obtained by conditioning $\D'$ on the sample space 
$\powerset_{ \fin }( \Facts )$: Intuitively, given that we know that a database 
instance was one of the original instances of the uncompleted PDB, it's 
probability stays the same. Note that there is no connection between
\cref{def:completion} and the notion of completion from measure theory. The
homonymy is a mere coincidence.

\begin{proposition}\label{pro:sp_completion}
	Let $\D_1$ and $\D_2$ be two PDBs with $\Facts( \D_1 ) \cap \Facts( \D_2 ) =
	\emptyset$ and $\Facts( \D_1 ) \cup \Facts( \D_2 ) = \Facts[ \tau, \UU ]$.
	Then $\D_1 \uplus \D_2$ is a completion of $\D_1$ if and only if\/ $\Pr_{ D
	\sim \D_2 }\event{ D = \emptyset } > 0$.
\end{proposition}

\begin{proof}
	For $i = 1, 2$ we let $\D_i = ( \DB_i, P_i )$ and $\Facts_i = 
	\Facts( \D_i )$ such that $\Facts_1$ and $\Facts_2$ form a partition of the
	set $\Facts[\tau, \UU]$ of all facts. 
	
	First suppose that $\complete{ \D } = \D_1 \uplus \D_2 = ( \complete{ \DB },
	\complete{ P } )$ is a completion of $\D_1$. By construction, 
	\[
		0 
		< \complete{ P }\big( \DB_1 \big)
		= P_1 \big( \DB_1 \big) \cdot P_2 \big( \set{ \emptyset } \big)
		= P_2 \big( \set{ \emptyset } \big)
		\text.
	\]

	For the backwards direction suppose $P_2\big( \set{ \emptyset } \big) > 0$
	and again, let $\D_1 \uplus \D_2 = \complete{ \D } = ( \complete{ \DB }, 
	\complete{ P } )$. Then $\hat{P}( \DB_1 ) = P_1( \DB_1 ) \cdot P_2\big(
	\set{ \emptyset } \big) = P_2\big( \set{ \emptyset } \big) > 0$ and for all
	$D \in \DB_1$ it holds that
	\[
		\complete{ P } \big( \set{ D } \under \DB_1 \big)
		= \frac{ \complete{ P } \big( \set{ D } \cap \DB_1 \big) }
			{ \complete{ P } \big( \DB_1 \big) }
		= \frac{ P_1 \big( \set{ D } \big) \cdot 
				P_2 \big( \set{ \emptyset } \big) }
			{ P_1 \big( \DB_1 \big) \cdot P_2 \big( \set{ \emptyset } \big) }
		= P_1\big( \set{ D } \big)\text.
	\]
	That is, $\complete{ \D }$ is a completion of $\D_1$.
\end{proof}

It is easy to see that not every completion of a PDB $\D$ can be written as a
superposition of $\D$ with a PDB of a fact set disjoint to that of $\D$. In 
fact, this is already the case for finite PDBs.

\begin{example}
	Suppose $\Facts[\tau, \UU] = \set{ f_1, f_2 }$. Let $\D = ( \DB, P )$
	with $\DB = \set{ \emptyset, \set{ f_1 } }$ and $P\big( \set{ f_1 } \big) = 
	p \in (0,1]$. Consider the PDB $\complete{ D } = ( \complete{ \DB }, 
	\complete{ P } )$ with $\complete{ \DB } = \powerset( \set{ f_1, f_2 } )$
	and $\complete{ P }$ according to the following table (where $\complete{ p }
	\in (0, 1)$):

	\begin{figure}[H]
		\centering%
		\begin{tabular}{ c cc cc }
			\toprule
			$D$
						& $\emptyset$
									& $\set{ f_1 }$
												& $\set{ f_2 }$	
															& $\set{ f_1, f_2 }$\\
			\midrule
			$\complete{ P }\big( \set{ D } \big)$
						& $(1 - p)(1 - \complete{p})$	
									& $p(1-\complete{p})$
												& $0$	
															& $\complete{p}$\\
			\bottomrule
		\end{tabular}
	\end{figure}
       
	Observe that  $\complete{ P }( \DB ) = 1 - \complete{p}$ and $\complete
	P(\{D\})=P(\{D\})\cdot(1-\complete p)$ for $D\in\DB$. Thus for $D\in\DB$ we
	have
	\begin{equation*}
		\complete{ P }\big( \set{ D } \mid \DB \big)
		= \frac{\complete P(\set D)}{ \strut\complete{P} (\DB) } 
		= \frac{ P\big( \set{ D }\big) \cdot (1 - \complete{p}) } 
			{ 1 - \complete{ p } }
		= P \big( \set{ D } \big)
		\text.
	\end{equation*}

	Hence $\complete{ \D }$ is a completion of $\D$. However, for any
	superposition $\D \uplus \D' = ( \complete{ \DB }, \complete{ P }' )$ where
	$\D'$ is any PDB with $\Facts( \D') = \set{ f_2 }$, it holds that
	$\complete{ P }' \big( \set{ f_2 } \big) > 0$ whenever $\complete{ P }' 
	\big( \set{ f_1, f_2 } \big) > 0$.  Thus, the completion $\complete{ \D }$
	considered above can not be expressed as a superposition of PDBs with
	disjoint fact sets.
\end{example}

In the light of \cref{pro:sp_completion}, a straightforward recipe for building
completions is to superpose a given PDB with a $\TI$-PDB that is given by
prescribed marginals for the remaining facts.

\begin{definition}[Independent Completions]\label{def:ind_completion}
	Let $\D$ be a PDB. A completion $\complete{ \D }$ of $\D$ is a
	\emph{$\TI$-completion} if $\complete{ \D } = \D \uplus \D'$ for some
	$\TI$-PDB $\D'$ with the property that for all facts $f \in
	\Facts[\tau,\UU]$ it holds that 
	\begin{equation}\label{eq:compl_facts}
		\Pr_{ D' \sim \D' } \event{ f \in D' } > 0
		\Rightarrow
		\Pr_{ D \sim \D } \event{ f \in D } = 0
	\end{equation}
	for all facts $f \in \Facts[\tau,\UU]$.

	Similarly, a completion $\complete{ \D }$ of $\D$ is a
	\emph{$\BID$-completion} if $\complete{ \D } = \D \uplus \D'$ for some
	$\BID$-PDB satisfying \eqref{eq:compl_facts}.
\end{definition}

Note that $\complete{ \D }$ will most likely not share the independence 
properties of $\D'$. However, we intuitively keep the independencies from $\D'$
through \cref{lem:sp_ind}. For example, this applies to the independence of the 
new facts in the completion. Moreover, if $\D$ is a $\TI$-PDB (resp. a
$\BID$-PDB) then $\complete{ \D }$ is a $\TI$-PDB (resp. a $\BID$-PDB) as well.

In \cite{Ceylan+2016}, the authors consider representations of finite 
$\TI$-PDBs that are given as a list of pairs $\smash{\big( f, P(f) \big)_{ f 
\in \Facts } }$ as is common practice in the treatment of finite PDBs 
\cite{Suciu+2011,VanDenBroeckSuciu2017}. According to \cite{Ceylan+2016}, 
representing PDBs this way inherently comes with a \emph{closed-world 
assumption} \cite{Reiter1981} though: If $\D_{ \spn{ \Facts, P } }$ is the 
$\TI$-PDB spanned by the list $\big( f, P( f ) \big)_{ f \in \Facts }$, we do
not model the uncertainty of facts $f \in \Facts[\tau, \UU] \setminus \Facts$
which for probabilistic query evaluation is equivalent to treating them as
events of probability $0$. As Ceylan~et~al. proceed to argue
\cite{Ceylan+2016}, this has several undesired practical consequences when
working with such representations.  They propose a model of
\emph{$\lambda$-completions} consisting of finite completions of a given finite
$\TI$-PDB by $\TI$-PDBs modeling the probabilities of all remaining facts. In
the following, we translate their construction into our framework.

\begin{definition}[Open-World Probabilistic Databases {\cite[Definition 4 \& 5]{Ceylan+2016}}]
	\label{def:lambda_completion}
	Let $\UU$ be finite. An \emph{open probabilistic database} is a pair
	$\G = ( \D, \lambda )$ where $\D = ( \DB, P )$ is a $\TI$-PDB over some 
	$\Facts \subseteq \Facts[\tau, \UU]$. A \emph{$\lambda$-completion} of $\D$
	is a $\TI$-PDB $\complete{ \D } = ( \complete{ \DB }, \complete{ P } )$ with
	$\complete{ \DB } = \powerset\big( \Facts[ \tau, \UU ] \big)$,
	\[
		\begin{aligned}[b]
		\complete{ P }( f ) &= P( f ) 
								  && \text{for all } f \in \Facts\text{ and}\\
		\complete{ P }( f ) &\leq \lambda 
								  && \text{for all } f \in \Facts[ \tau, \UU ]
								  		\setminus \Facts\text.
		\end{aligned}
		\qedhere
	\]
\end{definition}

Clearly, every ($\lambda$-)completion in the sense of Ceylan~et~al.
(\cref{def:lambda_completion}) is a completion in the sense of
\cref{def:completion}. That is, every $\lambda$-completion $\complete{ \D }$
can be written as
\[
	\complete{ \D } 
	= \D \uplus \D_{ \spn{ \Facts[\tau,\UU]\setminus\Facts, \complete{P} } }
\]
for the corresponding assignment $\complete{P} \from \Facts[\tau,\UU]
\setminus \Facts \to [0,\lambda]$ of marginal probabilities to the new facts.

The set of all $\lambda$-completions for some fixed $\lambda$ is the central
object of study in \cite{Ceylan+2016}. With our framework, the idea gently 
generalizes to infinite completions.

\begin{remark}
	The model of independent completions we described is, to our knowledge, not
	(or at least not directly) expressible in the existing \enquote{infinite}
	PDB systems \cite{KennedyKoch2010,Singh+2008a,AgrawalWidom2009} but could be
	emulated in MCDB \cite{Jampani+2011} as follows: First create $n$ dummy
	tuples where $n$ is drawn from a Poisson-distribution. For each of the $n$
	dummy tuples, sample their attribute entries independently according to a
	common probability measure on the space of facts using MCDBs \emph{Variable
	Generating functions}. According to \cite[Proposition
	3.5]{LastPenrose2017}, this describes a Poisson process on the space of
	facts, so it has the desired independence condition, but may contain
	duplicates (see also \cref{s:bag_instances,s:beyond}). We expect though,
	that a considerable speedup over this approach could be achieved in an
	sampling based implementation that directly exploits the independence
	properties.
\end{remark}

\subsection{Approximate Query Evaluation}
\label{ss:approx_eval}

Query evaluation in infinite PDBs is not the main focus of this paper and
remains an object of future study for the most part. Nevertheless, we present
two first results on query evaluation in infinite tuple-independent PDBs 
highlighting the bounds of possibility as well as connections to query
evaluation in finite (tuple-independent) PDBs. For query evaluation in finite
tuple-independent PDBs, the PDB can be given as part of the input by just
specifying the list of facts together with their marginal probabilities, as
seen in \cref{ex:order}. For analyzing complexity it is typically assumed that
all occurring marginal probabilities (and thus, in the finite, all instance
probabilities) are rational \cite{Gradel+1998,DalviSuciu2012}. As our PDBs can
be countably infinite, we need to comment on the data model. The basic
assumption replacing the exhausting list of fact probabilities is that given a
fact, we can determine its marginal probability in the input PDB. 

The first result we give states that we can compute additive approximations of 
the probability of Boolean query in countably infinite $\TI$-PDBs. We first
introduce this statement with respect to an oracle mechanism for accessing fact
probabilities, implying that this result is independent of the concrete
representation.  

Let $\D = (\DB, P)$ be a PDB. We say an algorithm \emph{has oracle access to
$\D$} (cf. \cite[Section 3.4]{AroraBarak2009}), if it can query a black box
\begin{enumerate}
	\item to obtain $P( f )$ given a fact $f$; and
	\item to obtain $\sum_{ f \in \Facts(\D) } P( f )$.
\end{enumerate}
In employing such an oracle mechanism we avoid the discussion of representation
issues at this point.

\begin{proposition}\label{pro:approx-additive}
	There is an algorithm $\A$ that, given oracle access to a $\TI$-PDB  $\D$,
	an $\epsilon>0$, and a Boolean query $\varphi \in \FO[\tau, \UU]$, returns a
	rational number $p( \phi )$ such that
	\[
		\Pr_{ D \sim \D }\big( D \models \phi \big) - \epsilon 
		\leq p( \phi ) \leq 
		\Pr_{ D \sim \D }\big( D \models \phi \big) + \epsilon
		\text.\qedhere
	\]
\end{proposition}

\begin{remark}
	We note that the proof does not actually rely on $\phi$ being an 
	$\FO$-query. All we need is that it comes from a class of queries such that
	there exists an algorithm that given a finite PDB and a Boolean query from
	the class returns the probability of the query being true.
\end{remark}

 \begin{proof}
	Let $\Facts( \D ) = \set{ f_1, f_2, \dots }$ and let $p_i = P( f_i )$.
	For all $n \in \NN_{ > 0 }$ let $\DB_n = \powerset\big( \set{ f_1, \dots,
	f_n } \big)$ and let $P_n \coloneqq P \restriction \set{ f_1, \dots, f_n }$.
	By \cref{lem:restr_ti}, $\D \under \DB_n = \D_{ \spn{ \set{ f_1, \dots, f_n 
	}, P_n } }$. Because $\D \under \DB_n$ is a \emph{finite} $\TI$-PDB (and we
	have access to its marginal probabilities using the oracle), the exact value
	of $P( \phi \under \DB_n )$ can be found using the traditional techniques
	for query answering in finite $\TI$-PDBs. 

	Since	$\D$ is a $\TI$-PDB, by \cref{thm:ti_series} it holds that 
	$\sum_{ i = 1 }^{ \infty } p_i < \infty$, so $\lim_{ i \to \infty } p_i = 
	0$. For all $n \in \NN_{ > 0 }$ we define
	\[
		r_n \coloneqq \sum_{ i = n + 1 }^{ \infty } p_i
		\text.
	\]
	Then $\lim_{ n \to \infty } r_n = 0$. We choose $n$ such that $r_n \leq 
	\epsilon$. A suitable $n$ can be computed by systematically listing facts
	$f_1, \dots, f_n$ until $r_n$ is small enough. The value $r_n$ itself can
	be calculated as $r_n = \sum_{ i = 1 }^{ \infty } p_i - \sum_{ i = 1 }^{ n } 
	p_i$ using the oracle access.

	Let $\DB_n \coloneqq \powerset\big( \set{ f_1, \dots, f_n } \big)$. We let
	our algorithm return $\A( \phi ) = p( \phi ) \coloneqq P( \phi \under \DB_n
	)$. Note that
	\[
		P( \DB_n )
		= 		\prod_{ i = n + 1 }^{ \infty } ( 1 - p_i )
		\overset{\eqref{eq:prodsum}}{\geq}	
				1 - \sum_{ i = n + 1 }^{ \infty } p_i
		= 		1 - r_n
		\geq	1 - \epsilon\text.
	\]
	Thus,
	\[
		P( \phi )
		= 	\underbrace{ P( \phi \under \DB_n ) }_{ = \A( \phi ) }
				\cdot \underbrace{ P( \DB_n ) }_{ \leq 1 }
			+	\underbrace{ P\big( \phi \under ( \DB_n )^{ \c } \big) }_{ \leq 1 } 
				\cdot \underbrace{ P\big( ( \DB_n )^{ \c } \big) }_{ \leq\epsilon } 
		\leq	\A( \phi ) + \epsilon\text,
	\]
	so $\A( \phi ) \geq P( \phi ) - \epsilon$. Also, we have
	\[
		P( \phi )
		=	\underbrace{ P( \phi \under \DB_n ) }_{ = \A( \phi ) }
				\cdot \underbrace{ P( \DB_n ) }_{ \geq 1 - \epsilon }
			+	\underbrace{ P\big( \phi \under ( \DB_n )^{ \c } \big) }_{ \geq 0 }
				\cdot \underbrace{ P\big( ( \DB_n )^{ \c } \big) }_{ \geq 0 }
		\geq	\A( \phi ) ( 1 - \epsilon )
		\geq	\A( \phi ) - \epsilon\text,
	\]
	so $\A( \phi ) \leq P( \phi ) + \epsilon$. Together, we have
	\[
		P( \phi ) - \epsilon \leq \A( \phi) \leq P( \phi ) + \epsilon 
	\]
	as required.
\end{proof}

In the problem discussed before, the database was fixed in the algorithm
whereas the query was the input. Next we also want to consider PDBs as an 
input.

\begin{definition}
	Let $\D = \D_{ \spn{ \Facts[\tau, \Sigma^*], P } }$ be a $\TI$-PDB where
	$P \from \Facts[\tau, \Sigma^*] \to [0, 1] \cap \QQ$. Let $M$ be a Turing
	machine with input alphabet $\Sigma \cup \tau \cup \set{ \mathord(,
	\mathord) }$ and let $\xi \in \QQ_{\geq 0}$.
	The pair $(M, \xi)$ \emph{represents} $\D$ if $M$ computes the function $p_M
	\from \Facts[\tau, \Sigma^*] \to [0, 1] \cap \QQ \with f \mapsto P( f )$ and
	$\xi = \sum_{ f \in \Facts[\tau, \Sigma^*] } P( f )$.
\end{definition}

The general computational problem we are interested in is then the following:

\begin{table}[H]
	\centering%
	\begin{tabular}{ l p{.75\textwidth} }
		\toprule
		\multicolumn{2}{l}{Probabilistic Query Evaluation $\PQE$}\\
		\midrule
		\textbf{Input} &
		A pair $(M_{\D}, \xi_{\D})$ representing a $\TI$-PDB $\D$ over $\tau$ and
		$\Sigma^*$ and a Boolean query $\phi \in \FO[\tau,\Sigma^*]$. \\
		\textbf{Output} & 
		The probability of $\phi$ being satisfied in $D \sim \D$, that is, $\Pr_{
		D \sim \D }\big( D \models \phi \big)$.\\
		\bottomrule
	\end{tabular}
\end{table}

We denote by $\PQE(\phi)$ the above problem with the input query being fixed to
some Boolean query $\phi \in \FO[\tau, \Sigma^*]$.

\begin{corollary}\label{cor:approx-additive}
	For all $\epsilon > 0$ there exists an additive $\epsilon$-approximation 
	algorithm for $\PQE$.
\end{corollary}

\begin{proof}
	Said algorithm can be obtained by following the proof of 
	\cref{pro:approx-additive} and replacing calls to the oracle by calculations
	using the representation of $\D$.
\end{proof}

\cref{pro:approx-additive,cor:approx-additive} give us additive approximation
results for query answering in countable $\TI$-PDBs. We note that for
approximation guarantee $\epsilon$ the run-time of the algorithm used in the
proof depends on the rate of convergence of the series of fact probabilities.
This is because the number $n$ of facts that are taken into consideration has
to be chosen such that $r_n$ (the sum of the remaining probabilities) is at
most $\epsilon$. In the best case, the facts are ordered in decreasing 
probability and $n$ is chosen minimal. Then, for example if the series of fact
probabilities is a geometric series, it holds that $n = \Omega\big(
\frac{ 1 }{ 1- \epsilon } \big)$.\footnote{Note though, that in general, series
may converge \enquote{arbitrarily slow}, see \cite[pp. 310--311]{Knopp1996}.}
The run-time of the complete algorithm is determined by the run-time of the
method that is used for the finite query evaluation on a $\TI$-PDB with $n$
facts.

\medskip

The next proposition shows that a multiplicative approximation algorithm does
not exist by investigating the data complexity of a particular (and very
simple) query.

\begin{proposition}\label{pro:approx-multiplicative}
	Let $\Sigma = \set{ 0, 1 }$ and $\tau = \set{ R, S }$ with $R$ and $S$ 
	unary.
	Let $\phi = \exists x \: R(x) \in \FO[\tau, \Sigma^*]$ and let $\rho \in
	\RR$, $\rho \geq 1$. Then there is no algorithm $\A$ for $\PQE(\phi)$ that
	on input a representation $(M_{\D},\xi_{\D})$ of a $\TI$-PDB $\D$ satisfies
	\begin{equation*}
		\rho^{ -1 } \Pr_{ D \sim \D }\big( D \models \phi \big) 
		\leq \A\big( M_{\D}, \xi_{\D} \big) \leq 
		\rho \Pr_{ D \sim \D }\big( D \models \phi \big) 
		\text.
		\qedhere
	\end{equation*}
\end{proposition}

\begin{proof}
	We let
	\begin{equation*}
		\decode{\?} \from \Sigma^* \to \NN_{ > 0 } \with
		w_1 \dots w_n \mapsto \sum_{ i = 1 }^{ n } w_i 2^{ i - 1} + 2^{ n }
	\end{equation*}
	be the bijection that identifies a $\Sigma$-string $w$ with the positive
	integer whose binary representation is $1w$. Note that $\decode{\?}$ is
	computable.
	
	For a Turing machine $M$ over alphabet $\Sigma = \set{ 0, 1 }$ we let $L_M$
	denote the set of strings in $\Sigma^*$ that are accepted by $M$. By Rice's 
	Theorem \cite{Rice1953} (cf. \cite[Theorem 34.1]{Kozen1997}), the
	set $\EMPTY$, that is, the set of (encodings of) Turing machines $M$
	with $L_M = \emptyset$, is undecidable. For every $t \in \NN_{ > 0 }$, we
	let $L_{M,t} = \set{ w \in \Sigma^* \with M \text{ accepts } w \text{ in}
	\leq t \text{ steps} }$. Then $L_{M, t}$ is clearly decidable for all $t
	\in \NN_{ > 0 }$ and it holds that $L_M = \bigcup_{ t \in \NN_{ > 0 } }
	L_{M, t}$.

	Let $M$ be a Turing machine over $\Sigma$. We define a $\TI$-PDB $\D$
	over $\Facts[\tau, \Sigma^*]$. Let $\pair{ \?, \? }$ be the function from
	$\NN_{ > 0 } \times \NN_{ > 0 }$ to $\NN_{ > 0 }$ with
	\begin{equation*}
		\pair{i,j} \coloneqq \binom{ i + j - 1 }{ 2 } + i
		= \frac12 \big( i + j - 1 \big)\big( i + j - 2 \big) + i
	\end{equation*}
	for all $i, j \in \NN_{ > 0 }$. It is well known, that $\pair{ \?, \? }$ is
	a computable bijection (see, for example \cite[Example J.2]{Kozen1997}). The
	marginal probabilities of $\D$ are defined as follows for all $w \in 
	\Sigma^*$:
	\begin{align}
		P\big( R( w ) \big) &= 
		\begin{cases}
			2^{ -\decode{w} }	
			& \text{if } \decode{w} = \pair{ n,t } \text{ and } n \in L_{M,t}
			\text{ and}\\
			0			
			& \text{otherwise.}
		\end{cases}
		\label{eq:EMPTY-R}\\
		P\big( S( w ) \big) &=
		\begin{cases}
			2^{ -\decode{w} } 
			& \text{if } \decode{w} = \pair{ n,t }\text{ and } n \notin L_{M,t}
			\text{ and}\\
			0			
			& \text{otherwise.}
		\end{cases}
		\label{eq:EMPTY-S}
	\end{align}
	Note that with these definitions, it holds that $\sum_{ f \in \Facts[\tau,
	\Sigma^*] } P( f ) = \sum_{ k \in \NN_{ > 0 } } 2^{ -k } = 1 < \infty$.
	Thus, $\D = \D_{ \spn{ \Facts[\tau,\Sigma^*], P } }$ is a well-defined PDB.
	(In particular note that the tuples in $S$ are used to bring the sum of
	marginal probabilities to $1$.)

	We observe that
	\begin{equation}
		\Pr_{ D \sim \D }\big( R(w) \in D \big) = 0\quad
		\text{for all }w \in \Sigma^*
		\qquad\Leftrightarrow\qquad
		\smashoperator{\bigcup_{ t \in \NN_{ > 0 } }}\mkern6mu L_{M,t} = 
		L_M = \emptyset\text.
	\end{equation}
	Recalling that $\phi = \exists x\: R(x)$, the above equivalence entails 
	that
	\begin{equation}
		\Pr_{ D \sim \D }\big( D \models \phi \big) =0
		\qquad\Leftrightarrow\qquad
		L_M = \emptyset\text.
	\end{equation}

	We construct a Turing machine $\widetilde{M}_{\D}$ over the alphabet
	$\widetilde \Sigma = \Sigma \cup \tau \cup \set{ \mathord(, \mathord) }$
	that works as follows:
	\begin{itemize}
		\item On input $w \in \widetilde{\Sigma}^*$, $\widetilde{M}_{\D}$ checks
			whether $w \in \Facts[\tau, \Sigma^*]$. If not, it rejects.
		\item Otherwise, it checks which of the cases in
			\cref{eq:EMPTY-R,eq:EMPTY-S} applies and outputs $2^{-\decode{w}}$ or
			$0$ accordingly.
	\end{itemize}
	Then $(\widetilde{M}_{\D},1)$ represents $\D$. 

	Now suppose that $\A$ is a multiplicative approximation algorithm for 
	$\PQE(\phi)$. Then 
	\begin{equation*}
		\A\big( \widetilde{M}_{\D}, 1 \big) = 0
		\quad\Leftrightarrow\quad
		\Pr_{ D \sim \D }\big( D \models \phi \big) = 0
		\quad\Leftrightarrow\quad 
		L_M = \emptyset
		\text.
	\end{equation*}
	Thus, $\A$ can be used to decide $\EMPTY$.
\end{proof}

From \cref{pro:approx-multiplicative} we immediately obtain the following
corollary.

\begin{corollary}
	Let $\rho \geq 1$. Then there exists no multiplicative $\rho$-approximation
	algorithm for $\PQE$.
\end{corollary}

\subsection{Tuple-Independent Bag PDBs}
\label{s:bag_instances}

In this last subsection on countable tuple-independent PDBs, we study
countable PDBs with a bag semantics.
Still maintaining assumptions \ref{ass:ctbl1} and \ref{ass:ctbl2}, we replace assumption
\ref{ass:ctbl3} by the following.
\begin{enumerate}[label=(\Roman*)]
\item[(III')]\label[assumption]{ass:ctbl3'} 
		If not explicitly stated otherwise, all PDBs that occur in this 
		section have sample space $\powerbag_{ \fin }\big(
      \Facts( \D ) \big)$ and are equipped with the powerset $\sigma$-algebra. 
\end{enumerate}
Independent PDBs with a bag semantics are interesting in their own
right, but the following discussion also serves as a preparation for our
treatment of uncountable PDBs in the next section. For uncountable
PDBs, it is easier to work with a bag semantics, simply because the
underlying probability theory of finite point processes usually
has been developed for point processes where points may be repeated.

\begin{definition}[Countable $\TI$-PDBs, Bag Version]\label{def:bag_ti_pdb}
	Let $\D$ in $\ctblPDB$. Then $\D$ is called \emph{tuple-independent} (or,
	a \emph{$\TI$-PDB}) if the numbers of occurrences of distinct facts are 
	mutually independent. That is, $\D \in \ctblPDB$ is a $\TI$-PDB if (and only
	if)
	\begin{equation*}
		\Pr_{ D \sim \D } 
			\event[\big]{ \mult_D (f_1) = n_1, \dots, \mult_D(f_k) = n_k }
		= \Pr_{ D \sim \D } \event[\big]{ \mult_D (f_1) = n_1 } 
			\cdot \dotsc \cdot 
			\Pr_{ D \sim \D } \event[\big]{ \mult_D (f_k) = n_k }
	\end{equation*}
	for all pairwise distinct $f_1, \dots, f_k \in \Facts( \D )$, all
	$n_1,\dots,n_k \in \NN$, and all $k \in \NN$. 

	The class of countable tuple-independent (bag) PDBs is denoted by 
	$\ctblTI$ and the subclass of finite tuple-independent (bag) PDBs by 
	$\finTI$.
\end{definition}

If we treat sets as bags with multiplicities in $\set{0, 1}$, then
\cref{def:bag_ti_pdb} is compatible with \cref{def:ti_pdb}. That is,
$\ctblsetTI \subseteq \ctblTI$ and, in particular, $\finsetTI \subseteq 
\finTI$. Similar to set PDBs, we can express tuple-independent bag PDBs as
superpositions of single-fact PDBs.

\begin{corollary}\label{cor:bagPDBsuperposition}\leavevmode
  \begin{enumerate}
  \item
    Let $\D$ be a countable $\TI$-PDB. For every $f \in \Facts$, let
	 $\D_f$ be the bag PDB with $\Facts(\D_f) = \set{ f }$, and with the 
	 probability measure defined by $\Pr_{ D \in \D_f }\big( \mult_D( f ) = k
	 \big) \coloneqq \Pr_{ D \in \D } \big( \mult_D( f ) = k\big)$.
    Then
    \[
      \D=\biguplus_{f\in\Facts(\D)}\D_f.
    \]
    Moreover, $\sum_{f\in\Facts(\D)}\Pr_{D\in\D}\big(\mult_D(f)>0\big)$ is finite.
  \item
    Let $\Facts\subseteq\Facts[\tau,\UU]$, and for every $f\in\Facts$,
    let $\D_f$ be a PDB with $\Facts(\D_f)=\{f\}$. Suppose that
	 $\sum_{f\in\Facts}\Pr_{D\sim\D_f}\big(\mult_D(f)>0\big)<\infty$.
    Then $\biguplus_{f\in\Facts}\D_f$ is a tuple-independent PDB.
	 \qedhere
  \end{enumerate}
\end{corollary}

\begin{example}
  A (countable) \emph{Poisson PDB} is a tuple-independent PDB
  $\D$ where the fact multiplicities are Poisson distributed,
  that is, for every $f\in\Facts(\D)$ there is a nonnegative
  $\lambda_f\in\RR$ such that 
  \[
    \Pr_{D\sim\D}\big(\mult_D(f)=k\big)=e^{-\lambda_f}\frac{{\lambda_f}^k}{k!}
    \text.
  \]
  Our definition of tuple-independence only requires that the
  multiplicities of distinct facts be independent; it makes no
  assumptions on the distributions of the individual fact
  multiplicities. 
  Intuitively, we may argue that Poisson PDBs also make an independence
  assumption for these distributions. Indeed, we might think of a
  tuple-independent bag
  PDB as being generated by sampling independent identical copies of
  each fact. Then if we have $n$ identical copies of each fact, each
  with a probability $p_{f}^{(n)}$, the
  fact multiplicities will be binomially distributed. As we let the
  number $n$ of copies go to infinity while keeping the expected value
  $\lambda_f=np_{f}^{(n)}$
  of the number of samples of each fact (and hence the expected size of the
  PDB) constant, these binomial distributions converge to a Poisson
  distribution with parameter $\lambda_f$ (see \cite[Section 6.5]{Feller1968}).

  As we will see in the next section, Poisson PDBs also play a special
  role in the theory of uncountable tuple-independent PDBs.
\end{example}

\begin{remark}
	A special special case of particular interest is given when the 
	total number of facts of a tuple-independent bag PDB is finite, but their
	individual multiplicities are unbounded. We call this a \emph{fact-finite
	PDB}. Fact-finite PDBs occupy a middle ground, as although there are only
	finitely many different facts, the sample space can be of infinite size.
	Answering queries in fact-finite (bag) PDBs is studied in \cite{Grohe+2022}.
\end{remark}

\section{Beyond Countable Domains}
\label{s:beyond}

In this section, we discuss a suitable notion of tuple-independence for PDBs
with uncountable sample spaces. This is an application of the theory of point
processes and of random measures, and builds on some more background from
measure theory and general topology that can be looked up in
\cref{app:prob,app:topo} whenever necessary.

We use the framework of \emph{standard PDBs} from \cite{GroheLindner2022}. From
this point of view, a probabilistic database is nothing but a \emph{finite
point process} \cite{DaleyVere-Jones2003}. Essentially, a finite point process
is a probability distribution over finite sets or bags of elements
(\enquote{points}) in some measurable space.  In the case of PDBs, these points
are the facts, and the number of times a particular fact occurs in an outcome
of the point process gives its multiplicity in the corresponding database
instance. 

In the following, let $\UU$ be an uncountable universe. Following
\cite{GroheLindner2022}, we require that the universe is given as a
\emph{standard Borel space} $(\UU,\UUU)$. That is, $\UU$ is a Polish
topological space with Borel $\sigma$-algebra $\UUU$. Given a database schema
$\tau$, this induces a natural $\sigma$-algebra $\FFF$ on the space
$\FF[\tau,\UU]$ of $(\tau,\UU)$-facts, which in turn generates a natural
$\sigma$-algebra for probabilistic databases through the events
\[
		\set[\big]{ D \in \DB[\tau,\UU] \with
			\mult_{ D }( \Facts_1 ) = n_1,
			\dotsc,
			\mult_{ D }( \Facts_k ) = n_k
		}
\]
for measurable sets $\Facts_1, \dots, \Facts_k$ of facts, and non-negative
integers $n_1,\dots,n_k$. We denote this $\sigma$-algebra by
$\DDD_{\mult}[\tau,\UU]$. For further information on the model and the detailed
constructions, we refer the reader to \cite{GroheLindner2022}. 

It is worth noting that, as soon as we move to uncountable spaces of facts, the
\emph{measurability} of constructions, functions and queries is a key issue
that needs to be addressed. The model from \cite{GroheLindner2022}, however,
has been shown to exhibit the desired properties for probabilistic databases,
such as the measurability of typical database queries. That is, the model
itself, and the semantics of queries are well-defined.

\begin{definition}[Standard PDBs, \cite{GroheLindner2022}]\label{def:stdPDB}
	Let $\tau$ be a database schema and let $(\UU, \UUU)$ be a standard Borel 
	universe. A \emph{standard PDB} over $\tau$ and $\UU$ is a probability space
	$\D = ( \DB, \DDD, P )$ with
	\begin{itemize}
		\item sample space $\DB = \DB[ \tau, \UU ] = \powerbag_{ \fin }\big(
			\Facts[ \tau, \UU ] \big)$ and
		\item $\sigma$-algebra $\DDD = \DDD_{ \mult } \big( \Facts[ \tau, \UU ]
			\big)$.\qedhere
	\end{itemize}
\end{definition}

In probability theoretic terms, \cref{def:stdPDB} is just the definition of a
finite point process over the adequately constructed measurable space of facts
$\FF[\tau,\UU]$. In turn, finite point processes are special \emph{random
measures} \cite{Kallenberg2017}, namely, random \emph{integer-valued} measures.

Recall that under the tuple-independence assumption, the presence (or
multiplicities) of pairwise distinct facts are independent. Given a continuum
of possible facts, this definition is too weak, as it fails to capture
independence between \enquote{regions} of the fact space. This was no issue in
countable PDBs as there, each of the countably many possible instances can be
expressed as an intersection of the marginal events from the definition of
tuple-independence. Here, however, these marginal events alone do not suffice
to describe the measurable structure of PDBs.

Independence has been investigated thoroughly in the general theory of
random measures, though. A random measure is called \emph{completely random}
\cite{Kingman1967}, if its values on any finite number of disjoint measurable
subsets of the space are independent. Translating this to the language of PDBs
gives a direct generalization of the tuple-independence assumption for PDBs
over continuous spaces.

\begin{definition}[Tuple-Independence for Standard PDBs, cf.\ \cite{Kingman1967}]
	\label{def:standard_ti_pdb}%
	A standard PDB $\D = ( \DB, \DDD_{\mult}, P )$ is called 
	\emph{tuple-independent} (or, a \emph{$\TI$-PDB}) if for all $k = 1, 2, 
	\dots$, all $n_1, \dots, n_k = 0, 1, 2, \dots$ and all mutually disjoint
	$\Facts_1, \dots, \Facts_k \in \FFF$ it holds that
	\begin{equation}
		\Pr_{ D \sim \D }
		\set[\big]{
			\mult_{ D }( \Facts_1 ) = n_1,
			\dotsc,
			\mult_{ D }( \Facts_k ) = n_k
		}
		= \prod_{ i = 1 }^{ k }
			\Pr_{ D \sim \D } \set[\big]{ \mult_{ D }( \Facts_i ) = n_i }
		\text.\qedhere
		\label{eq:standardTI}
	\end{equation}
\end{definition}

As singletons $\set{f}$ are measurable in $\Facts[ \tau, \UU ]$ by the
construction of \cite{GroheLindner2022}, the above is an extension of the
notion of tuple-independence we introduced earlier for countable PDBs.
Moreover, for countable PDBs, \cref{def:standard_ti_pdb} is equivalent to
\cref{def:bag_ti_pdb}, and for countable set PDBs, it is equivalent to
\cref{def:ti_pdb}.

\smallskip

The structure of completely random measures is well-understood. A classic 
result due to Kingman \cite{Kingman1967} presents a decomposition of completely
random measures into three well-structured parts by means of
superposition. Here, the term superposition refers to the sum of random 
measures (usually implying that they are independent). This is a direct
generalization of our notion of superposition for countable PDBs from
\cref{s:countable}.

For \emph{integer-valued} completely random measures, Kingmans decomposition
simplifies as follows.
\begin{fact}[see~{\cite[Theorem~2.4.VI]{DaleyVere-Jones2003}}]\label{fac:newdecomp}
	Every integer-valued completely random measure $\mu$ is a superposition of
	two random measures $\mu_1$ and $\mu_2$, where $\mu_1$ is completely random
	with countable support, and $\mu_2$ is a random measure defined by a 
	\emph{compound Poisson process} satisfying $\Pr( \mu_2( \set{x} ) > 0 ) =
	0$ for all $x$).\footnote{In particular, the diffuse deterministic
	component of Kingmans characterization \cite[§~8]{Kingman1967} vanishes
	when going from general random measures to integer-valued ones (cf. also
	\cite[Proposition~9.1.III(i-ii)]{DaleyVere-Jones2008}.} 
\end{fact}

To make clear the implications for PDBs with independence assumptions, let us
expand a bit more on the involved terminology. A \emph{Poisson process}
\cite{LastPenrose2017} on a standard Borel space $( \Omega, \mathfrak A )$ is a
point process (i.e., integer-valued random measure) that is parameterized
through a measure $\lambda$ on $( \Omega, \mathfrak A )$ such that
\begin{enumerate}
	\item the number of points in every $A \in \mathfrak A$ is
		Poisson-distributed with parameter $\lambda(A)$, and
	\item the numbers of points in disjoint measurable sets are independent (as
		in \cref{def:standard_ti_pdb}).
\end{enumerate}
A compound Poisson process can be thought of as a generalization, specifying
random locations of points by the means of a Poisson process, and for these
points, prescribing separate independent, positive multiplicity distributions.
For the precise statements and further details, we refer to the literature,
specifically \cite[Chapter 2]{DaleyVere-Jones2003} and \cite[Chapter 9 \&{}
10]{DaleyVere-Jones2008}.

While the characterization from \cref{fac:newdecomp} highlights a strong
connection between independence assumptions and the Poisson process, Poisson
processes themselves also yield a very simple model for uncountable PDBs.

\begin{definition}[Poisson-PDBs]\label{def:poisson-pdb}
  A PDB $\D = (\DB, \DDD, P)$ is called a \emph{Poisson-PDB} if there exists a
  finite measure $\lambda$ on $( \Facts[ \tau, \UU ], \FFF )$, called the 
  \emph{parameter (measure)} of $\D$, such that
  \begin{enumerate}
  \item $\D \in \standardTI$; and
  \item for all $\Facts \in \FFF$ we have
    \begin{equation}\label{eq:poisson-fact-set}
      \Pr_{ D \sim \D } \event[\big]{ \mult_D( \Facts ) = k }
      = \frac{ \lambda( \Facts )^k }{ k! } \cdot e^{ -\lambda( \Facts ) }
      \text.
    \qedhere
    \end{equation}
       \label{itm:poisson-pdb-ii}
  \end{enumerate}
\end{definition}

That is, in a Poisson-PDB, the random variable $\mult_{( \? )}( \Facts ) \from
\DB \to \NN_{ > 0 } \with D \mapsto \mult_D( \Facts )$ is Poisson-distributed
with parameter $\lambda( \Facts )$. We emphasize again that the parameter
$\lambda$ is not a single number, but rather a function (more precisely, a
measure) that maps every measurable set of facts to a non-negative real number.
We note that Poisson-PDBs over the fact space $( \Facts[\tau,\UU], \FFF )$
exist for every choice of parameter measure $\lambda$ (see
\cite[Theorem~3.6]{LastPenrose2017}).

Now from \cref{fac:newdecomp}, we obtain the following characterization of
$\standardTI$-PDBs.

\begin{theorem}
	Let $\D = ( \DB, \DDD_{\mult}, P )$ be a $\standardTI$-PDB over $\tau$ and
	$\UU$. Then $\D$ is a superposition of two PDBs $\D_1$ and $\D_2$ such that
	\begin{itemize}
		\item $\D_1$ is a countable $TI$-PDB, and 
		\item the deduplication of $\D_2$ is a Poisson-PDB with parameter
			$\lambda$ satisfying $\lambda(\set{f}) = 0$ for all $f \in
			\Facts[\tau,\UU]$.\qedhere
	\end{itemize}
\end{theorem}

In particular, if $\D$ is a $\standardTI$-PDB whose instances are almost surely
set instances (that is, if $\D$ is simple), and $\Pr\big( \mult(\set{f}) > 0
\big) = 0$, it follows that $\D$ is a Poisson-PDB
\cite[see~Theorems~6.9~and~6.12]{LastPenrose2017}.

In fact, it already follows that $\D$ is a Poisson-PDB, if $\D$ is any simple
standard PDB, and for which there exists a diffuse, finite measure $\lambda$ on
$(\Facts[\tau,\UU],\FFF)$ such that $P\big( \mult( \Facts ) = 0 \big) =
e^{-\lambda(\Facts)}$ for all measurable sets $\Facts$ of facts \cite[Theorem
6.10]{LastPenrose2017}. Then $\lambda$ is the parameter measure of the PDB.

Apart from their significance for the tuple-independence assumption,
Poisson-PDBs have some additional nice properties (see
\cite[Theorems~3.3~and~5.2]{LastPenrose2017}).  For example, every
superposition of two Poisson-PDBs, say, with parameters $\lambda_1$ and
$\lambda_2$ is a Poisson-PDB with parameter $\lambda_1 + \lambda_2$. Moreover,
the restriction of a Poisson-PDB to a smaller, measurable set of facts is again
a Poisson-PDB.

While the above unravels the notion of tuple-independence for PDBs over
uncountable fact spaces, it may not be clear where to go from here. The
probabilistic tools we have touched are used in a plethora of models, for
example in ecology, epidemiology and astronomy \cite{Baddeley2007}, and our
point of view suggests that such models can be treated as probabilistic
databases. A particular application we see is thus the extension of existing
data by such a model (in the guise of an uncountable PDB), which could pave the
way for a sophisticated variant of open-world query evaluation. Therefore, a
possible future research direction is the combination of techniques from point
process theory with query processing in PDBs.

\section{Concluding Remarks}
\label{s:conclusion}

We introduce a formal framework of infinite probabilistic databases. Within the
framework, we study tuple-independence and related independence assumptions.  
This adds to the theoretical foundation of existing PDB systems that support
infinite domains in their data model and opens various directions for future
research.

We show that countable tuple-independent PDBs exist exactly for convergent
series of marginal fact probabilities. From a more abstract view, PDBs with
independent components can be explained using the notion of superpositions, a
known concept from point process theory. Towards this end, we investigate some
general properties of superpositions of PDBs, most notably, how they preserve
independence and when they indeed yield valid PDBs as a result. Following this
approach, it turns out that (countable) tuple-independent PDBs can be
decomposed into arbitrary smaller PDBs. The modularity of the superposition
approach can also be used to reason about block-independent disjoint PDBs and
possibly more general classes that are obtained by closing a subclass of PDBs
under superpositions. Superpositions also enable us to define tuple-independent
PDBs with a bag semantics in a natural way, leading us to the notion of 
Poisson-PDBs.

The vast increase in expressive power by allowing infinite probability spaces
comes at a cost, though. We show that in this setting, and contrary to the
finite situation, there are (countable) PDBs that can not be expressed as a
first order view of a tuple-independent PDB. Yet, we argue that a simple
tuple-independent model of completions in the sense of \cite{Ceylan+2016} can
be used to obtain more meaningful query results in (finite) PDBs. While we
can't even hope for multiplicative approximation guarantees in infinite
tuple-independent PDBs, query evaluation in such PDBs can be additively
approximated using the well-established methods for probabilistic query
evaluation in finite PDBs. 

Key problems for future research in infinite PDBs are accessible (finite)
representations of infinite PDBs and query evaluation algorithms.
Representations of countable PDBs as views over tuple-independent PDBs have
been studied in \cite{Carmeli+2021}. Our results about (approximate) query
evaluation in infinite PDBs are only a first step, and an in-depth
investigation is still open. There are some natural follow-up questions
regarding our results in \cref{ss:approx_eval}, for example, what could be said
about the query evaluation problem for restricted classes of PDBs. One aspect
we deem particularly interesting is discussing the query evaluation problem in
the bag semantics setup for fact-finite Poisson-PDBs. For putting query
evaluation in infinite PDBs into practice (beyond the state of affairs that
has been pointed out in the related work section), a promising approach seems to
try to integrate traditional database techniques with probabilistic inference
techniques for infinite domains that are used in AI, including, for example the
relational languages and models BLOG \cite{Milch+2005,Wu+2018}, ProbLog
\cite{DeRaedt+2007,Gutmann+2011} and Markov Logic
\cite{RichardsonDomingos2006,SinglaDomingos2007}. In general, the underlying
inference problems are of high complexity, so achieving tractability is
challenging. 
\subsection*{Acknowledgments}
We wish to express our gratitude to Christoph Standke and Anton Pirogov for 
their interest, their many suggestions and their opinion that helped shape the 
presentation of this work.

This work is funded by \href{https://www.dfg.de}{\emph{Deutsche
Forschungsgemeinschaft} (DFG, German Research Foundation)} under grants
\href{https://gepris.dfg.de/gepris/projekt/412400621}{GR 1492/16-1} and
\href{https://gepris.dfg.de/gepris/projekt/282652900}{GRK 2236 (UnRAVeL)}.

\appendix

\section{Mathematical Background}

\subsection{Series and Products}\label{app:infsumprod}

We use \cite{Knopp1996} as our standard reference regarding the theory of
(infinite) sums and products. For the readers convenience, this section recaps
basic definitions and well-known results about series and infinite products
that are used throughout the article.\par

Let $\big( a_i \big)_{ i \in \NN }$ be a sequence of real numbers $a_i \in 
\RR$. The formal expression $\sum_{ i = 0 }^{ \infty } a_i$ is called a 
\emph{series}. If the limit $\lim_{ n \to \infty } \sum_{ i = 0 }^{ n } a_i$
exists, and is equal to $a \in \RR$, then $\sum_{ i = 0 }^{ \infty }
a_i$ is called \emph{convergent}, $a$ is called its \emph{value} and we write
$\sum_{ i = 0 }^{ \infty } a_i = a$. If the limit is exists and is $\infty$,
then $\sum_{ i = 0 }^{ \infty } a_i$ is said to \emph{diverge to $\infty$} and
we write $\sum_{ i = 0 }^{ \infty } a_i = \infty$.  Every series we deal with
in this article has non-negative terms only.  Note that any such series is
either convergent, or diverges to $\infty$.  A series $\sum_{ i = 0 }^{ \infty
} a_i$ is called \emph{absolutely convergent} if $\sum_{ i = 0 }^{ \infty }
\abs{ a_i }$ converges. Obviously, every convergent series with only
non-negative terms is absolutely convergent. 

\begin{fact}[{\cite[ch. IV, Theorem 1]{Knopp1996}}] 
  If\/ $\sum_{ i = 0 }^{ \infty } a_i$ is absolutely convergent, then
  $\sum_{ i = 0 }^{ \infty } a_i = \sum_{ i = 0 }^{ \infty } \tilde a_i$ for
  every permutation $( \tilde a_i )_{ i \in \NN }$ of $( a_i )_{ i \in \NN }$.
\end{fact}

Since the sequences we consider in this article will be absolutely convergent
anyways, we will not have to worry about the order of summation. In particular,
we sum over unordered countable index sets.

Now let $\big( a_i \big)_{ i \in \NN }$ be a sequence of real numbers $a_i \in
\RR$. The formal expression $\prod_{i = 0}^{\infty} a_i$ is called an 
\emph{infinite product}. In this article, we will only deal with the case where
$a_i \in [0,1]$ for all $i \in \NN$. In this situation, the limit $\lim_{ n
\to \infty } \prod_{ i = 0 }^{ n } a_i$ always exists and is in the interval
$[0, 1]$, because the corresponding sequence of partial products monotonically
decreasing and bounded by $0$ from below. If the limit is $a \in [0,1]$, we
write $\prod_{ i = 0 }^{ \infty } a_i = a$. When discussing the relationship to
series, it is more convenient to write infinite products in the shape $\prod_{
i = 0 }^{ \infty } ( 1 - a_i )$.  It is a basic fact (see \cite[ch.~VII,
Theorems 7 and 11]{Knopp1996}) that for such infinite products, if $\sum_{ i =
0 }^{ \infty } a_i < \infty$, then $\prod_{ i = 0 }^{ \infty }( 1 - a_i ) =
\prod_{ i = 0 }^{ \infty } ( 1 - \tilde{a}_i )$ for every permutation $( \tilde
a_i )_{ i \in \NN }$ of $( a_i)_{ i \in \NN }$. This justifies writing infinite
products over unordered index sets. 

Moreover, the following connection between series and products holds in terms
of convergence. 

\begin{fact}[see {\cite[Theorems 125.1 and
	126.4]{Knopp1996}}]\label{fac:prodsumconv}
	Let $a_i \in [0,1]$ with $a_i \neq 1$ for all $i \in \NN$. Then
	$\sum_{i =0}^{\infty} a_i < \infty$ if and only if $\prod_{ i = 0 }^{\infty}
	(1-a_i) > 0$.
\end{fact}

Furthermore, for sequences $(a_i)_{ i \in
\NN }$ with $a_i \in [0,1]$, it holds that
\begin{equation}\label{eq:prodsum}
	\prod_{ i = 0 }^{ \infty } ( 1 - a_i )
	\geq
	1 - \sum_{ i = 0 }^{ \infty } a_i.
\end{equation}
This is a variant of the \emph{Weierstrass inequalities} \cite[p. 104 et 
seq.]{Bromwich1926}, and can easily be shown by an induction that considers the
partial products $\prod_{ i = 0 }^{ n } ( 1 - a_i )$.

\subsection{Probability Theory}\label{app:prob}

In this subsection we cover most of the relevant background from probability 
theory including some basic concepts from measure theory. We follow the 
textbooks \cite{Klenke2014,Kallenberg2002}, which the reader may consult as
needed for further information.

\subsubsection{Measurable spaces}

Let $\Omega \neq \emptyset$ be some set. A family $\AAA$ of subsets of $\Omega$
is called a \emph{$\sigma$-algebra} on $\Omega$ if
\begin{enumerate}
	\item $\Omega \in \AAA$,
	\item for all $\AAAA \in \AAA$ it holds that $\AAAA^{\c} = \Omega \setminus
		\AAAA \in \AAA$ (\emph{closure under complement}),
	\item for all $\AAAA_1, \AAAA_2, \dots \in \AAA$ it holds that 
		$\bigcup_{ i = 1 }^{ \infty } \AAAA_i \in \AAA$ (\emph{closure under 
		countable union}).
\end{enumerate}
It follows from the definition that if $\AAA$ is a $\sigma$-algebra on 
$\Omega$, it is also closed under countable intersection. That is, for all
$\AAAA_1, \AAAA_2, \dots \in \AAA$ it holds that $\bigcap_{ i = 1 }^{ \infty }
\AAAA_i \in \AAA$.

A pair $(\Omega, \AAA)$, where $\Omega \neq \emptyset$ and $\AAA$ is a 
$\sigma$-algebra on $\Omega$, is called a \emph{measurable space}. The elements
of $\AAA$ are called ($\AAA$-)\emph{measurable} sets. For every non-empty set 
$\Omega$, both $\powerset(\Omega)$ and $\set{ \emptyset, \Omega }$ are 
$\sigma$-algebras on $\Omega$.

Let $(\Omega_1, \AAA_1)$ and $(\Omega_2, \AAA_2)$ be measurable spaces. A
function $f \from \Omega_1 \to \Omega_2$ is called \emph{$(\AAA_1,
\AAA_2)$-measurable} if for all $\AAAA \in \AAA_2$ it holds that $f^{ - 1 }(
\AAAA ) \in \AAA_1$ where $f^{ -1 }( \AAAA ) = \set{ \omega \in \Omega \with
f(\omega) \in \AAAA }$. If $\AAA_1$ and $\AAA_2$ are clear from context, the
function is just called \emph{measurable}. 

Let $\GGG$ be a family of subsets of $\Omega \neq \emptyset$. Then 
$\sigma(\GGG)$ denotes the coarsest $\sigma$-algebra (that is, the smallest
one with respect to set inclusion) containing $\GGG$. Then $\sigma( \GGG )$ is
indeed unique and it holds that
\begin{equation*}
	\sigma( \GGG ) =
	\mkern10mu
	\smashoperator{\bigcap_{ \substack{ 
			\AAA \subseteq \powerset( \Omega )\text{ s.\,t.}\\
			\AAA \supseteq \GGG\text{ and}\\
			\AAA \text{ $\sigma$-algebra}
		}
	}}
	\mkern8mu
	\AAA
	\text.
\end{equation*}
We call $\sigma(\GGG)$ the $\sigma$-algebra \emph{generated by} $\GGG$.

\subsubsection{Measures}

Let $(\Omega, \AAA)$ be a measurable space. A function $\mu \from \AAA \to 
[0, \infty]$ is called a \emph{measure} on $(\Omega, \AAA)$  (or, on $\Omega$
if $\AAA$ is clear from context) if
\begin{enumerate}
	\item $\mu( \emptyset ) = 0$ and
	\item for all pairwise disjoint $\AAAA_1, \AAAA_2, \dotsc \in \AAA$ it holds
		that $\mu \big( \bigcup_{ i = 1 }^{ \infty } \AAAA_i \big) = 
		\sum_{ i = 1 }^{ \infty } \mu( \AAAA_i )$ (\emph{$\sigma$-additivity}).
		\qedhere
\end{enumerate}

If $\mu$ is a measure on a measurable space $(\Omega, \AAA)$, then $( \Omega,
\AAA, \mu )$ is called a \emph{measure space}. We call $\mu$ \emph{finite} if
$\mu( \Omega ) < \infty$ and a \emph{probability measure} if $\mu( \Omega ) = 
1$. If $\mu$ is a probability measure, then $( \Omega, \AAA, \mu )$ is called
a \emph{probability space}. In a probability space, measurable sets are also
called \emph{events} and for $\AAAA \in \AAA$, $\mu( \AAAA )$ is called the
\emph{probability} of $\AAAA$. We denote probability measures by $P$ instead of
$\mu$. 

\begin{fact}[see {\cite[Theorem 1.36]{Klenke2014}}]\label{fac:semicontinuity}
	Let $(\Omega, \AAA, \mu)$ be a measure space.
	\begin{enumerate}
		\item For all $\AAAA_1, \AAAA_2, \dotsc \in \AAA$ with 
			$\AAAA_1 \subseteq \AAAA_2 \subseteq \dotsb$ it holds that\/
			$
				\mu\big( \bigcup_{ i = 1 }^{ \infty } \AAAA_i \big) = 
				\lim_{ i \to \infty } \mu( \AAAA_i )
			$.
		\item For all $\AAAA_1, \AAAA_2, \dotsc \in \AAA$ with
			$\AAAA_1 \supseteq \AAAA_2 \supseteq \dotsb$ and with 
			$\mu( \AAAA_i ) = \infty$ and for at most finitely many $i \in \NN_{
			>0 }$ it holds that\/ 
			$
				\mu\big( \bigcap_{ i = 1 }^{ \infty } \AAAA_i \big) = 
				\lim_{ i \to \infty } \mu( \AAAA_i )
			$.\qedhere
	\end{enumerate}
\end{fact}

If $( \Omega_1, \AAA_1, \mu )$ is a measure space, $( \Omega_2, \AAA_2 )$ a
measurable space, then every measurable function $f \from \Omega_1 \to 
\Omega_2$ induces a measure $\mu_2$ on $( \Omega_2, \AAA_2 )$ via
\begin{equation*}
	\mu_2( \AAAA ) 
	= \mu_1\big( \set{ \omega \in \Omega_1 \with f( \omega ) \in \AAAA } \big)
	\text.
\end{equation*}
Then $\mu_2$ is called the \emph{image} or \emph{push-forward measure} of
$\mu_1$ under $f$ and $( \Omega_2, \AAA_2, \mu_2 )$ is called \emph{image
measure space}. If $\mu_1$ is a probability measure, so is $\mu_2$. In this
situation, $f$ is called a \emph{random variable}.

\subsubsection{Stochastic independence}

Let $( \Omega, \AAA, P )$ be a probability space and let $I$ be some non-empty
index set. A family of events $\big( \AAAA_i \big)_{ i \in I }$ with $\AAAA_i
\in \AAA$ for all $i \in I$ is called \emph{(stochastically) independent} if
for all $k = 1, 2, \dots$ and all pairwise different $i_1, \dots, i_k \in I$ 
it holds that
\begin{equation}\label{eq:stochastic_independence}
	P \bigg( \bigcap_{ j = 1 }^{ k } \AAAA_{i_j} \bigg)
	= \prod_{ j = 1 }^{ k } P \big( \AAAA_{i_j} \big)
	\text.
\end{equation}
The family $\big( \AAAA_i \big)_{ i \in I }$ is called \emph{pairwise
independent} if \eqref{eq:stochastic_independence} holds for $k = 2$ and $i_1,
i_2 \in I$ with $i_1 \neq i_2$.

\begin{fact}[see {\protect\cite[Theorem 2.5]{Klenke2014}}]%
	\label{fac:exchangeindependent}%
	Let $\big( \AAAA_i \big)_{ i \in I }$ and $\big( \bmskew\widetilde{ \AAAA 
	}_i \big)_{ i \in I }$ be families of events in a probability space 
	$(\Omega, \AAA, P)$ where $\bmskew\widetilde{ \AAAA }_i \in \set{ \AAAA_i,
	{\AAAA_i}^{\c} }$ for all $i \in I$. Then\/ $\big( \AAAA_i \big)_{ i \in I 
	}$ is independent if and only if $\big( \bmskew\widetilde{ \AAAA }_i 
\big)_{ i \in I }$
	is independent.
\end{fact}

\subsubsection{Product measure spaces}

Let $\big((\Omega_i, \AAA_i)\big)_{ i \in I }$ be a family of measurable spaces
for some non-empty index set $I$. The \emph{product $\sigma$-algebra} 
$\bigotimes_{ i \in I } \AAA_i$ of the $\AAA_i$, $i \in I$ is the coarsest 
$\sigma$-algebra on $\Omega \coloneqq \prod_{ i \in I } \Omega_i$ making all
the \emph{canonical projections} $\pi_i \from \Omega \to \Omega_i \with 
( \omega_i )_{ i \in I } \mapsto \omega_i$ measurable. That is,
\begin{equation}\label{eq:product_sigma_algebra}
	\bigotimes_{ i \in I } \AAA_i
	\coloneqq \sigma\big( \set{ \pi_i^{-1}( \AAAA_i ) \with \AAAA_i \in \AAA_i }
	\big)
	\text.
\end{equation}
If $I = \set{ 1, \dots, n }$, we write $\bigotimes_{ i = 1 }^{ n } \AAA_i$ or
$\AAA_1 \otimes \dots \otimes \AAA_n$, and if $I = \NN_{ > 0 }$, we write 
$\bigotimes_{ i = 1 }^{ \infty } \AAA_i$ instead of $\bigotimes_{ i \in I }
\AAA_i$. With \eqref{eq:product_sigma_algebra}, $\bigotimes_{ i \in I } (
\Omega_i, \AAA_i ) \coloneqq (\Omega, \bigotimes_{ i \in I } \AAA_i )$ is a
measurable space, and is called the \emph{product measurable space} of the $(
\Omega_i, \AAA_i)$, $i \in I$.

\begin{fact}[{see \cite[Corollary 14.33]{Klenke2014}}]%
	\label{fac:product_measure}%
	Let $( \Omega_i, \AAA_i, P_i )$ be probability spaces for $i = 1, 2, \dots$.
	Moreover, let $\Omega \coloneqq \prod_{ i = 1 }^{ \infty } \Omega_i$ and
	$\AAA \coloneqq \bigotimes_{ i = 1 }^{ \infty } \AAA_i$.
	Then there exists a unique probability measure $P$ on $\big( 
	\prod_{ i = 1 }^{ \infty } \Omega_i, \bigotimes_{ i = 1 }^{ \infty } \AAA_i
	\big)$ with the property that
	\begin{equation*}
		P\bigg( \bigcap_{ j = 1 }^{ k } 
			\pi_{ i_j }^{ -1 } \big( \AAAA_{ i_j } \big) \bigg)
		= \prod_{ j = 1 }^{ k } P_{ i_j }\big( \AAAA_{ i_j } \big)
	\end{equation*}
	for all $k = 1, 2, \dots$, all pairwise distinct $i_1, \dots, i_k \in 
	\NN_{ > 0 }$ and all $\AAAA_{i_j} \in \AAA_{i_j}$.
\end{fact}

The probability space $( \Omega, \AAA, P )$ from \cref{fac:product_measure} is
called the \emph{product probability space} of the spaces $( \Omega_i, \AAA_i,
P_i )$ and $P$ is called the associated \emph{product probability measure} of
the $P_i$, $i = 1, 2, \dots$.

\begin{fact}\label{fac:product_ind}
	Let $( \Omega, \AAA, P )$ be the product probability space from 
	\cref{fac:product_measure} and let $\AAAA_i \in \AAA_i$ for all $i = 1, 2, 
	\dots$. Then $\big( \pi_i^{ -1 } ( \AAAA_i ) \big)_{ i \in \NN_{ > 0 } }$ is
	independent in $( \Omega, \AAA, P )$.
\end{fact}

\subsection{Some Background from Topology}\label{app:topo}
This subsection contains some notions from general and metric topology that are
relevant for the definitions in \cref{s:beyond}. For a general reference, 
consult \cite{Willard2004}. For standard Borel spaces and their properties see
\cite[Section 424]{Fremlin2013}.

A \emph{topological space} is a pair $(X, \TTT)$ where $X$ is a set and $\TTT$
is a family of subsets of $X$ such that
\begin{itemize}
	\item $\emptyset, X \in \TTT$,
	\item $\TTT$ is closed under finite intersections, i.\,e. if $X_i \in \TTT$
		for all $i \in I$ where $I$ is some \emph{finite} index set, then 
		$\bigcap_{i \in I} X_i \in \TTT$, and
	\item $\TTT$ is closed under arbitrary unions, i.\,e. if $X_i \in \TTT$ for
		all $i \in I$ where $I$ is an \emph{arbitrary} index set, then 
		$\bigcup_{i\in I} X_i \in \TTT$.
\end{itemize}
The family $\TTT$ is then called a \emph{topology} on $X$ and its elements
are called the \emph{open sets} of $(X,\TTT)$. Complements of open sets are
called \emph{closed}. If $\TTT$ is clear from context, we just call $X$ a 
topological space, referring to $(X, \TTT)$.

A \emph{metric space} is a pair $(X,d)$ where $X$ is a set and $d \from X
\times X \to \RR_{\geq 0}$ is a \emph{metric} on $X$. That is,
\begin{itemize}
	\item $d(x,y) \geq 0$ with $d(x,y) = 0$ if and only if $x = y$,
	\item $d(x,y) = d(y,x)$ and
	\item $d(x,y) \leq d(x,z) + d(z,y)$
\end{itemize}
for all $x, y, z \in  X$.\par

Let $(X,d)$ be a metric space, $x \in X$ and $\epsilon > 0$. Then
$B_{\epsilon}(x) \coloneqq \set{ y \in X \with d(x,y) < \epsilon }$ 
denotes the \emph{open ball of radius $\epsilon$ around $x$}. The \emph{metric
topology} $\TTT(d)$ on $X$ with respect to $d$ is the topology whose open sets
are exactly the sets $Y \subseteq X$ with the property that for all $y \in Y$
there exists $\epsilon > 0$ such that $B_{\epsilon}(y) \subseteq Y$. This is
indeed a topology, and every open set is a union of open balls as defined
before.

A \emph{Cauchy sequence} in a metric space $(X,d)$ is a sequence of elements 
$x_1, x_2, \dotsc \in X$ such that $\lim_{n \to \infty} d(x_n,x_{n+1}) = 0$. A 
sequence $x_1, x_2, \dotsc \in X$ \emph{converges} in $(X,d)$, if there exists 
$x \in X$ such that for all $\epsilon > 0$ there exists $n \in \NN$ such that 
$d(x, x_n) < \epsilon$. A metric $d$ on $X$ is called \emph{complete}, if all
Cauchy sequences in $(X, d)$ converge.

A topological space $(X, \TTT)$ is called \emph{metrizable}, if there exists a
metric $d$ on $X$ such that $\TTT = \TTT(d)$. The space is called 
\emph{completely} metrizable, if $d$ can be chosen to be complete.

A topological space $(X, \TTT)$ is called \emph{separable}, if it contains a
countable dense subset, that is, a countable set $D \in \TTT$ such that 
$\overline{D} = \bigcap \set{ Y \subseteq X \with Y \text{ closed and } D 
\subseteq Y } = X$. (The set $\overline{D}$ is called the \emph{closure} of 
$D$.)

A topological space is \emph{Polish}, if it is completely metrizable and 
separable. A measurable space $(\Omega, \AAA)$ is \emph{standard Borel} if
$\AAA$ is generated by the open sets of a Polish topology on $\Omega$.

\bibliographystyle{plainurl}

\end{document}